\documentclass[onecolumn,draftclsnofoot]{IEEEtran}
\usepackage{amsmath}
\usepackage{graphicx}
\usepackage{caption2}
\usepackage{amsthm}
\usepackage{float}
\usepackage{verbatim}
\usepackage{epstopdf}
\usepackage{mathtools}
\usepackage{mathrsfs}
\usepackage{amssymb}
\usepackage{amsmath}
\usepackage{amsthm}
\usepackage{amsfonts}
\usepackage{subfigure}
\usepackage{color}
\usepackage{cite}
\usepackage{cancel}
\usepackage{dsfont}
\usepackage[shortlabels]{enumitem}
\usepackage[pagebackref=true,breaklinks=true,letterpaper=true,colorlinks,bookmarks=false]{hyperref}
\usepackage{algorithm}
\usepackage{algpseudocode}
\usepackage{arydshln}
\usepackage[official]{eurosym}

\DeclareMathOperator*{\argmax}{arg\,max}
\newcommand\independent{\protect\mathpalette{\protect\independenT}{\perp}}
\def\independenT#1#2{\mathrel{\rlap{$#1#2$}\mkern2mu{#1#2}}}

\newtheorem{theorem}{Theorem}
\newtheorem{proposition}{Proposition}
\newtheorem{remark}{Remark}
\newtheorem{corollary}{Corollary}[theorem]
\newtheorem{example}{Example}
\newtheorem{lemma}{Lemma}
\newtheorem{definition}{Definition}

\usepackage{tikz}
\usetikzlibrary{dsp,chains}
\usepackage{pgfplots}
\pgfplotsset{compat = newest}
\usetikzlibrary{shapes.misc, positioning}
\usetikzlibrary{pgfplots.external}
\usetikzlibrary{external}

\pgfplotscreateplotcyclelist{nano mark style}{%
dashed, every mark/.append style={solid, fill=gray!30}, mark=*, line width=1.1pt\\%
dashed, every mark/.append style={solid, fill=gray!30}, mark=square*, line width=1.1pt\\%
dashed, every mark/.append style={solid, fill=gray!30}, mark=triangle*, mark size=2.5pt, line width=1.1pt\\%
dashed, every mark/.append style={solid, fill=gray!30},mark=diamond*, mark size=2.5pt, line width=1.1pt\\%
dashed, every mark/.append style={solid, fill=gray!30},mark=x, mark size=3pt, line width=1.1pt\\%
}

\pgfplotscreateplotcyclelist{nano line style}{%
color = black, line width=2pt\\%
densely dashed, color = black, line width=2pt\\%
densely dotted, color = black, line width=2pt\\%
loosely dashed, color = black, line width=2pt\\%
}

\begin{document}
\title{Data Disclosure under Perfect Sample Privacy}
\author{Borzoo Rassouli$^1$, Fernando E. Rosas$^2$, and Deniz G\"und\"uz$^2$\\
$^1$ University of Essex, Colchester, UK\\
$^2$ Imperial College London, London, UK\\
{\tt\small b.rassouli@essex.ac.uk}, {\tt\small  f.rosas@imperial.ac.uk}, {\tt\small  d.gunduz@imperial.ac.uk}
}
 
\maketitle
\begin{abstract}
Perfect data privacy seems to be in fundamental opposition to the economical and scientific opportunities associated with extensive data exchange. Defying this intuition, this paper develops a framework that allows the disclosure of collective properties of datasets without compromising the privacy of individual data samples. 
We present an algorithm to build an optimal disclosure strategy/mapping, and discuss it fundamental limits on finite and asymptotically large datasets. Furthermore, we present explicit expressions to the asymptotic performance of this scheme in some scenarios, and study cases where our approach attains maximal efficiency. We finally discuss suboptimal schemes to provide sample privacy guarantees to large datasets with a reduced computational cost.
\end{abstract}
\begin{IEEEkeywords}
Data disclosure, inference attacks, data privacy, latent features, perfect privacy
\end{IEEEkeywords}


\section{Introduction}

\subsection{Context}

The fundamental tension between the benefits of information exchange and the need of data privacy is at the heart of the digital society. On the one hand, the massive amount of available data is currently enabling important scientific and economic opportunities; for example, experimental data can nowadays be shared effortlessly between researchers to allow parallel analyses, and consumer preferences can be extracted from online activity to aid the design of new products and services. On the other hand,  recent cases of misuse of user data (e.g. the well-known cases of Facebook and Cambridge Analytica~\cite{cadwalladr2018revealed}) are raising major concerns about data privacy, which has become a preeminent topic with overarching social, legal, and business consequences. As a matter of fact, important efforts have been taken to guarantee user privacy, including the \textit{General Data Protection Regulation} (GDPR) adopted by the European Union at an estimated cost of \euro{} 200 billion \cite{Georgiev2019}, and the recent adoption of differential privacy~\cite{dwork2014algorithmic} standards by major tech companies including Apple and Google. A key open problem is how to satisfy sufficient privacy requirements while still enabling the benefits of extensive data sharing.

There have been important efforts to address this problem from academic and industrial sectors, which are mainly focused on developing privacy-preserving data processing techniques. Privacy-preserving data disclosure is based on the intuition that the content of a dataset can be divided in two qualitatively different parts: non-sensitive statistical regularities that exist across the data, and private information that refers to the contents of specific entries/users. This distinction suggests that -- at least in principle -- one could extract and share global properties of data, while keeping information about specific samples confidential.

The highest privacy standard that a data disclosure strategy can guarantee, called \textit{perfect privacy}, corresponds to when nothing can be learned about an individual that could not have been learned without the disclosed data anyway~\cite{dalenius1977towards}. While studied in~\cite{Calmon2,rassouli2017}, perfect privacy is often disregarded for being too restrictive, corresponding to a extreme choice within the trade-off that exists between privacy and utility \cite{li2009tradeoff,Basciftci2016}. The most popular approach that takes advantage of this trade-off is \textit{differential privacy}~\cite{dwork2006}, which is equipped with free parameters that can be flexibly tuned in order to adapt to the requirements of diverse scenarios. However, while these degrees of freedom provide significant flexibility, determining the range of values that can guarantee that the system is ``secure enough'' is usually not straightforward \cite{lee2011much}.

There is an urgent need of procedures that can enable effective data exchange while ensuring rigorous privacy guarantees. Within this context, the goal of this work is to revisit perfect privacy and present algorithms to build perfectly-private data disclosure procedures.

\subsection{Scenario and related work}\label{sec:1.2}

Let us consider a user who has a private dataset, denoted by $X^n\triangleq(X_1,\dots,X_n)$, which is correlated with a latent variable of interest, denoted by $W$, that the user would like to share with an analyst. Note that in this scenario the user has no direct access to $W$, but can only attempt to infer its contents via the information provided by $X^n$. For instance, $X^n$ can be measurements of a patient's vital signals while $W$ is a particular health indicator, e.g., the risk of heart attack. Although it would be desirable for the patient to share the whole dataset with a remote assessment unit to provide early alerts in case of an emergency, she may not want to disclose the data samples themselves as this could reveal unintended personal information.

We follow the framework for privacy against inference attacks~\cite{Monedero}, \cite{Calmon1}, which proposes to disclose a variable $Y$ that is obtained through a mapping from the data set. This work focuses on mappings that satisfy \textit{perfect sample privacy}; that is, mappings, whose output ($Y$) do not provide any useful information that could foster statistical inference on the value of any particular sample, i.e. on $X_i$ for all $i=1,\dots,n$. Mathematically, this is equivalent to consider only those mappings whereby $Y$ and $X_i$ are statistically independent for all $i=1,\dots,n$, while $W- X^n -Y$ form a Markov chain. To assess the quality of $Y$ as an estimator of $W$, we consider the mutual information between the two, $I(Y;W)$. This quantity is an adequate proxy -- with better algebraic properties -- for the estimation error rate (also known as \textit{0-1 loss})~\cite{hellman1970probability,feder1994relations}, which is a central performance metric for classification and many other machine learning tasks~\cite{shalev2014}.

It is important to note that the above conditions are not equivalent to imposing statistical independence between the disclosed variable $Y$ and the whole dataset $X^n$. In fact, if $X^n$ and $Y$ are independent then the data-processing inequality leads to $I(Y;W)\leq I(X^n;Y)=0$, implying that under this condition the analyst cannot receive any information about $W$. At this point, it is useful to recall a counterintuitive and largely underexploited feature of multivariate statistics: variables that are pairwise independent can still be globally interdependent \cite{jacod2012probability}. Said differently, while $I(Y;X^n)=0$ implies $I(Y;X_i)=0$ for $i=1,\dots,n$, the converse does not hold. For example, it is direct to verify that if $X_1$ and $X_2$ are two independent fair coins, then $Y= X_1\oplus X_2$ (i.e., their exclusive OR) is independent of each of them, while $I(X_1,X_2;Y)>0$ \cite{rosas2016understanding}. Therefore, in this case, $Y$ reveals a collective property (whether the entries of $(X_1,X_2)$ are equal or not), while saying nothing about the specific value of $X_1$ or $X_2$.

Differential privacy is driven by similar desiderata, but the corresponding set of techniques and guarantees are quite different. A disclosure mapping $Y$ is said to be $\epsilon$-differential private if $\log \mathbb{P}\{Y|X^n=x^n\} \leq \epsilon + \log \mathbb{P}\{Y|X^n=\hat{x}^n\}$ for any pair of datasets $x^n$ and $\hat{x}^n$ that differ in only one entry, so that the value of a particular sample is not supposed to affect too much the value of $Y$ and vice-versa. Hence, while perfect sample privacy guarantees strict independency between the disclosed data and each sample, differential privacy only limits their conditional dependency\footnote{A direct calculation shows that differential privacy imposes a restriction on the conditional mutual information, i.e. $I(Y;X_i|X_{-i}^n) \leq \epsilon$ for all $i=1,\dots,n$, where $X_{-i}^n$ stands for the whole dataset excluding $X_i$.}. Unfortunately, it has been shown that in some cases this latter restriction fails to provide privacy guarantees, as some $\epsilon$-private disclosure mechanisms can still allow the leakage of an unbounded amount of information -- independently of how small $\epsilon$ might be~\cite[Theorem 4]{Calmon1}.

Another related problem to the one considered here is the \textit{privacy funnel}, in which the goal is to reveal the data set $X$ within a given accuracy under some utility measure, while keeping the latent variable $W$ as private as possible \cite{Makhdoumi}. Also, various metrics for quantifying the quality of the disclosure strategy has been studied in \cite{Basciftci2016,rassouli2017,issa2016operational,borzoo2018}.

\subsection{Contributions}

In this paper we study a data disclosure technique that guarantees perfect sample privacy, which we call ``synergistic information disclosure'' as it reveals information about the whole dataset (i.e., about $X^n$), but not about any of its constituting elements (i.e., $X_i$'s). 
Building up on~\cite{rassouli2018latent}, we derive necessary and sufficient conditions that determine when information about a latent feature can be synergistically disclosed, and present a simple but tight upper bound on how much information can be disclosed. Moreover, we provide a practical procedure for building an optimal synergistic disclosure mapping, which is based on linear programming (LP) methods, as stated in Theorem \ref{theo1}. We illustrate this method on a simple scenario where the data set $X$ consists of two binary samples, for which we provide a closed-form expression for the performance of the optimal synergistic disclosure mapping. 

When considering large datasets, we obtain the asymptotic performance limit of optimal synergistic disclosure when the dataset is composed of noisy measurements of a phenomenon of interest in Theorems \ref{th4} and \ref{th5}. As a by-product of this analysis, we observe a link between the \textit{full data observation} and \textit{output perturbation} models in \cite{Ishwar}. Moreover, when considering self-disclosure, we show that in most cases the ratio of information that one can synergistically disclose about the dataset to the information content of the dataset tends asymptotically to one, provided in Theorem \ref{th2}. We also show, in Theorem \ref{th6}, that, when applied to datasets of continuous samples, the disclosure capacity diverges. Finally, we present two suboptimal schemes of low computational complexity to build disclosure mappings that still guarantee perfect sample privacy, which are well suited to large datasets composed of independent samples.


The rest of the paper is structured as follows. Section~\ref{sec:2} introduces the notion of perfect sample privacy, and develops the conditions and bounds that characterize the private disclosure capacity. Subsequently, Section~\ref{sec:3} proves that the optimal mapping can be found through an LP, and develops the case where the data set consists of two binary samples. Then, Section~\ref{sec:large_datasets} studies the asymptotic performance for datasets of noisy  observations of a latent feature, and Section~\ref{sec:synergistic_selfdisc} considers the limits of synergistic self-disclosure. Section~\ref{sec:cont} studies the case of datasets with continuous variables. Finally, Section~\ref{sec:conclusions} conveys our final remarks.

\subsection{Notation}

Random variables (r.v.'s) are denoted by capital letters and their realizations 
by lowercase letters. For two r.v.'s $X$ and $Y$, $X\independent Y$ indicates that they are statistically independent. Matrices and vectors are denoted by bold capital and bold lowercase letters, respectively. For a matrix $\mathbf{A}_{m\times k}$, the null space, rank, and nullity are denoted by $\mbox{Null}(\mathbf{A})$, $\mbox{rank}(\mathbf{A})$, and $\mbox{nul}(\mathbf{A})$, respectively, with $\mbox{rank}(\mathbf{A})+\mbox{nul}(\mathbf{A})=k$. For integers $m$ and $n$ such that $m\leq n$, we define the discrete interval $[m:n]\triangleq\{m, m+1,\ldots,n\}$, and 
for $[1:n]$, we use the shorthand notation $[n]$.
For an integer $n\geq 1$, $\mathbf{1}_n$ denotes an $n$-dimensional all-one column vector. For a finite alphabet $\mathcal{X}$, the probability simplex $\mathcal{P}(\mathcal{X})$ is the standard $(|\mathcal{X}|-1)$-simplex given by
\begin{equation*}
\mathcal{P}(\mathcal{X})=\bigg\{\mathbf{v}\in\mathbb{R}^{|\mathcal{X}|}\bigg|\mathbf{1}_{|\mathcal{X}|}^T\cdot\mathbf{v}=1,\ v_i\geq 0,\ \forall i\in [|\mathcal{X}|]\bigg\}.
\end{equation*}
To each probability mass function (pmf) on $\mathcal{X}$, denoted by $p_{X}(\cdot)$ (or written simply as $p_X$), corresponds a probability vector $\mathbf{p}_X\in \mathcal{P}(\mathcal{X})$, whose  $i$-th element is $p_X(x_i)$ ($i\in[|\mathcal{X}|]$). Likewise, for a pair of discrete r.v.'s $(X,Y)$ with joint pmf $p_{X,Y}$, the probability vector $\mathbf{p}_{X|y}$ corresponds to the conditional pmf $p_{X|Y}(\cdot|y),\forall y\in\mathcal{Y}$, and $\mathbf{P}_{X|Y}$ is an $|\mathcal{X}|\times|\mathcal{Y}|$ matrix with columns $\mathbf{p}_{X|y},\forall y\in\mathcal{Y}$.

\section{Definition and basic properties}
\label{sec:2}

\subsection{Perfect sample privacy and synergistic disclosure}

Consider a case where a user has access to a dataset, denoted by $X^n$, which is dependent on a latent variable of interest $W$ that the user wishes to share with an analyst. The constituting elements of the dataset, i.e., $X_i$'s (which are in general random variables), are informally referred to as ``data samples''. From a communication theoretic perspective, $X^n$ can be considered to be a set of (possibly noisy) observations of $W$. The variables $W,X^n$ are assumed to be distributed according to a given joint distribution $p_{W,X^n}$.

Our goal is to process the dataset $X^n$ in a way that the result is maximally informative about $W$, while keeping the content of each $X_i$ $\forall i\in[n]$ confidential. Without loss of generality, we consider data disclosure strategies that take the form of a stochastic mapping, which can be captured by a conditional pmf $p_{Y|X^n}$. By this construction, $W-X^n-Y$ form a Markov chain.
Throughout the paper, unless stated otherwise, we focus on the case where $|\mathcal{W}|,|\mathcal{X}_i|<\infty,\ \forall i\in[n]$.

Our first step is to provide a suitable definition of data privacy. Although the condition $Y\independent X^n$ is sufficient for guaranteeing perfect privacy of the data samples, it is too constrictive. In fact, from the data processing inequality, if $Y\independent X^n$ then $Y\independent W$, implying that such a requirement makes $Y$ useless for the task of informing about $W$. In the following, we introduce the notion of \emph{perfect sample privacy}, which imposes a set of more flexible constraints. 

\begin{definition}
Consider a stochastic mapping which is applied on a dataset $X^n$ and produces an output $Y$. This mapping guarantees \textit{perfect sample privacy} if and only if $Y\independent X_i$ $\forall i\in[n]$. Furthermore, the set of all such mappings is denoted by
\begin{equation}\label{def_set}
\mathcal{A} = \bigg\{ p_{Y|X^n} \:\bigg| \:Y\independent X_i,\forall i\in[n]\bigg\}.
\end{equation}
\end{definition}
Therefore, a variable $Y$ generated by processing $X^n$ via a mapping $p_{Y|X^n}\in\mathcal{A}$ cannot foster statistical inference attacks over any of the samples of the dataset. Interestingly, mappings that satisfy perfect sample privacy can still carry useful information about latent variables. Our key principle is \textit{synergistic disclosure}: that is possible for $Y$ to carry information about a latent feature $W$ while revealing no information about any of the individual data samples. The next definition formalizes this notion by characterizing the latent variables that are feasible of being synergistically disclosed.
\begin{definition}\label{def1}
For a given latent variable $W$ and dataset $X^n$, 
$W$ is said to be feasible of \textit{synergistic disclosure} if there exists a random variable $Y$ that satisfies the following conditions:
\begin{itemize}
\item[1.] $W-X^n-Y$ form a Markov chain,
\item[2.] $p_{Y|X^n} \in \mathcal{A}$,
\item[3.] $Y\not\independent W$.
\end{itemize}
Moreover, the \textit{synergistic disclosure capacity} is defined as
\begin{equation}\label{def}
I_s(W,X^n)\triangleq\sup_{\substack{p_{Y|X^n}\in \mathcal{A}:\\W-X^n-Y}} I(W;Y).
\end{equation}
Finally, the synergistic disclosure efficiency is defined as $\eta(W,X^n) \triangleq I_\text{s}(W,X^n)/H(W)$.
\end{definition}
\noindent

The term "synergistic" comes from the fact that a synergistic disclosure mapping reveals collective properties of the whole dataset that do not compromise its "parts" (i.e. the samples themselves). In the sequel, $I_\text{s}$ is employed as a shorthand notation for $I_\text{s}(W,X^n)$ when the dataset and latent feature are clear.

Let the support of $X^n$ be defined as
\begin{equation*}
\hat{\mathcal{X}} \triangleq \bigg\{x^n\in\prod_{i=1}^n \mathcal{X}_i\bigg|p_{X^n}(x^n)>0\bigg\}.
\end{equation*}
From this definition, $\mathbf{p}_{X^n}$ always lies in the interior of $\mathcal{P}(\hat{\mathcal{X}})$. Also, it is evident that $|\hat{\mathcal{X}}|\leq \Pi_{i=1}^{n}|\mathcal{X}_i|$. 

Define matrix $\mathbf{P}$ as 
\begin{equation}\label{matp}
\mathbf{P}\triangleq\begin{bmatrix}
\mathbf{P}_{X_1|X^n}\\ \vdots \\\mathbf{P}_{X_n|X^n}
\end{bmatrix}_{G\times|\hat{\mathcal{X}}|},
\end{equation}
where $G\triangleq \sum_{i=1}^n|\mathcal{X}_i|$. Note that $\mathbf{P}$ is a binary matrix, as $X_i$'s are deterministic functions of $X^n$. For example, if $|\mathcal{X}_i|=2, \forall i\in[n]$ and $\hat{\mathcal{X}}$ is the set of all binary $n$-sequences, i.e., $\hat{\mathcal{X}}=\{0,1\}^n$, then $\mathbf{P}$ is a $2n\times 2^n$ matrix that can be built recursively according to
\begin{equation}\label{matrix_structure}
\mathbf{P}_{k+1} = 
    \left[
    \begin{array}{c;{2pt/2pt}c}
    \vspace{-0.2cm}
    1 \dots 1 & 0 \dots 0 \\
    0 \dots 0 & 1 \dots 1 \\
    \hdashline[2pt/2pt]
    \mathbf{P}_k & \mathbf{P}_k
    \end{array}
    \right],  
    \nonumber
\end{equation}
with $\mathbf{P} = \mathbf{P}_n$ and $\mathbf{P}_1 = \Big[ \begin{array}{cc} \vspace{-0.3cm} 1 0 \\ 0 1\end{array}\Big] $. 
\begin{remark}(\textit{Graphical representation})
Let $V=\{v_i\}_{i=1}^G$ be a set whose elements are in a one-to-one correspondence with the realizations of the data samples as follows. The first $|\mathcal{X}_1|$ elements correspond to $\mathcal{X}_1$, the next $|\mathcal{X}_2|$ elements correspond to $\mathcal{X}_2$, and so on. It can be verified that to the joint distribution $p_{X^n}$, corresponds an n-uniform n-partite hypergraph $H_n=(V,E)$, in which, any hyperedge, i.e., any element of $E$, corresponds to an element of $\hat{\mathcal{X}}$, i.e., the support of $X^n$. Furthermore, the incidence matrix of $H_n$ is matrix $\mathbf{P}$ given in \eqref{matp}.  
\end{remark}
The importance of $\mathbf{P}$ is clarified in the following Lemma.
\begin{lemma}\label{lemma2}
We have $X_i\independent Y,\ \forall i\in[n]$, if and only if 
$(\mathbf{p}_{X^n}-\mathbf{p}_{X^n|y})\in\textnormal{Null}(\mathbf{P}),\forall y\in\mathcal{Y}$.
\end{lemma}
\begin{proof}
Let $X$, $Y$ and $Z$ be discrete r.v.'s that form a Markov chain as $X-Y-Z$. Having $X\independent Z$ is equivalent to $
p_{X}(\cdot)=p_{X|Z}(\cdot|z)$, i.e., $\mathbf{p}_{X}=\mathbf{p}_{X|z},\ \forall z\in\mathcal{Z}$.
Furthermore, due to the Markov chain assumption, we have $\mathbf{p}_{X|z} = \mathbf{P}_{X|Y}\mathbf{p}_{Y|z},\ \forall z\in\mathcal{Z}$, and in particular, $\mathbf{p}_{X}=\mathbf{P}_{X|Y}\mathbf{p}_{Y}$. Therefore, having $\mathbf{p}_{X}=\mathbf{p}_{X|z}$, $\forall z\in\mathcal{Z}$ results in
\begin{equation*}
\mathbf{P}_{X|Y} \left(\mathbf{p}_{Y}- \mathbf{p}_{Y|z} \right)=\mathbf{0},\:\forall z \in\mathcal{Z},
\end{equation*}
or equivalently, $\left(\mathbf{p}_{Y}-\mathbf{p}_{Y|z}\right)\in \text{Null}(\mathbf{P}_{X|Y})$, $\forall z\in\mathcal{Z}$.

The proof is complete by noting that i) $X_i-X^n-Y$ form a Markov chain for each index $i\in[n]$, and ii) $\mbox{Null}(\mathbf{P}) = \cap_{i=1}^n \mbox{Null}(\mathbf{P}_{X_i|X^n})$.
\end{proof}

\subsection{Fundamental properties of synergistic disclosure}

The following Proposition characterizes the class of features that are feasible of synergistic disclosure from a given dataset. 
\begin{proposition}\label{Prop1}
For a given pair $(W,X^n)$, $W$ is feasible of synergistic disclosure if and only if $\textnormal{Null}(\mathbf{P})\not\subset\textnormal{Null}(\mathbf{P}_{W|X^n})$. 
\end{proposition}
\begin{proof}
For the first direction, we proceed as follows. If $I_s>0$, we have $W\not\independent Y$ in $W-X^n-Y$. Therefore, there exist $y_1,y_2\in\mathcal{Y}$ , where $y_1\neq y_2$, such that $\mathbf{p}_{W|y_1}\neq \mathbf{p}_{W|y_2}$, and hence, $\mathbf{p}_{X^n|y_1}\neq \mathbf{p}_{X^n|y_2}$. Since $X_i\independent Y,\ \forall i\in[n]$, Lemma~\ref{lemma2} implies that $(\mathbf{p}_{X^n}-\mathbf{p}_{X^n|y_1}),(\mathbf{p}_{X^n}-\mathbf{p}_{X^n|y_2})\in \mbox{Null}(\mathbf{P})$. 
Moreover, $\mbox{Null}(\mathbf{P})\not\subset\mbox{Null}(\mathbf{P}_{W|X^n})$, since otherwise $\mathbf{P}_{W|X^n}(\mathbf{p}_{X^n}-\mathbf{p}_{X^n|y_1})=\mathbf{P}_{W|X^n}(\mathbf{p}_{X^n}-\mathbf{p}_{X^n|y_2})=\mathbf{0}$, which implies $\mathbf{p}_{W|y_1}=\mathbf{p}_{W|y_2}$ leading to a contradiction.

The second direction is proved as follows. If $\mbox{Null}(\mathbf{P})\not\subset\mbox{Null}(\mathbf{P}_{W|X^n})$, there exists a non-zero vector $\mathbf{v}\in\mbox{Null}(\mathbf{P})$, such that $\mathbf{v}\not\in \mbox{Null}(\mathbf{P}_{W|X^n})$. Let $\mathcal{Y}\triangleq\{y_1,y_2\}$, $Y\sim\mbox{Bern}(\frac{1}{2})$, and for sufficiently small $\epsilon>0$, let $\mathbf{p}_{X^n|y_i}\triangleq\mathbf{p}_{X^n}+(-1)^i\epsilon\mathbf{v}, i\in[2]$. This construction is possible as $\mathbf{p}_{X^n}$ lies in the interior of $\mathcal{P}(\hat{\mathcal{X}})$, and $\mathbf{1}_{|\hat{\mathcal{X}}|}^T\cdot\mathbf{v}=\mathbf{1}_{G}^T\cdot\mathbf{P}\mathbf{v}=0$, which follows from $\mathbf{1}_{|\hat{\mathcal{X}}|}^T=\mathbf{1}_{G}^T\cdot\mathbf{P}$, and $\mathbf{v}\in\mbox{Null}(\mathbf{P})$. Accordingly, since $\mathbf{p}_{X^n}-\mathbf{p}_{X^n|y_i}\in \text{Null}(\mathbf{P}),~i\in[2]$, from Lemma~\ref{lemma2}, we have $X_i\independent Y,\ \forall i\in[n]$. Also, in the construction of the pair $(X^n,Y)$, $\mathbf{p}_{X^n}$ is preserved, as specified in $p_{W,X^n}$. Therefore, we have $W-X^n-Y$. Finally, since $\mathbf{v}\not\in \mbox{Null}(\mathbf{P}_{W|X^n})$, from $\mathbf{p}_{W|y}=\mathbf{P}_{W|X^n}\mathbf{p}_{X^n|y}$, we get $\mathbf{p}_{W|y_1}\neq \mathbf{p}_{W|y_2}$, or equivalently, $I_s>0$.
\end{proof}

The characterization presented in Proposition~\ref{Prop1} can be understood intuitively as follows. Changing $\mathbf{p}_{X^n}$ along the vectors in $\textnormal{Null}(\mathbf{P})$ corresponds to conditional pmfs $p_{X^n|Y}$ whose corresponding $p_{Y|X^n}$ guarantee perfect sample privacy, while changing $\mathbf{p}_{X^n}$ along the vectors in $\textnormal{Null}(\mathbf{P}_{W|X^n})$ corresponds to the conditional pmfs $p_{X^n|Y}$ that result in $W\independent Y$. Therefore, the condition $\textnormal{Null}(\mathbf{P})\not\subset\textnormal{Null}(\mathbf{P}_{W|X^n})$ asks for the existence of conditional probabilities that guarantee perfect sample privacy while introducing statistical dependencies with $W$.

\begin{proposition}\label{pr:upper_bound}
The following upper bound holds for $I_s$:
\begin{equation}\label{uppbound}
I_s\leq \min_{j\in [n]} I(W;X_{-j}^n|X_j) ,
\end{equation}
where $X_{-j}^n \triangleq \{X_1,\ldots,X_n\}\backslash X_j$.
\end{proposition}

\begin{proof}
Let $j\in[n]$ be an arbitrary index. Then,
\begin{align}
I(W;Y)=&\:I(W;X^n)-I(W;X^n|Y)\label{e1}\\
=&\:I(W;X_{-j}^n|X_j)+I(W;X_j)-I(W;X_j|Y)-I(W;X_{-j}^n|X_j,Y)\nonumber\\
=&\:I(W;X_{-j}^n|X_j)+I(W;X_j)-I(W,Y;X_j)-I(W;X_{-j}^n|X_j,Y)\label{e2}\\
=&\:I(W;X_{-j}^n|X_j)-I(Y;X_j|W)-I(W;X_{-j}^n|X_j,Y)\nonumber\\
\leq &\:I(W;X_{-j}^n|X_j),\label{fe2}
\end{align}
where (\ref{e1}) follows from the Markov chain $W-X^n-Y$, and (\ref{e2}) from the independence of $X_j$ and $Y$. Since $j$ is chosen arbitrarily, (\ref{fe2}) holds for all $j\in[n]$, resulting in (\ref{uppbound}).
\end{proof}
\begin{remark}
From Proposition~\ref{pr:upper_bound} and noting that $I(W;X^n_{-j}|X_j) = I(W;X^n) - I(W;X_j)$, one can find that in general
\begin{equation}
    I(W;X^n) - I_\text{s} \geq \max_{j\in[n]} I(W;X_j),
\end{equation}
which shows that amount of information that one needs to restrain from sharing for guaranteeing perfect sample privacy is at least equal to the amount of information carried by the most strongly correlated sample. 
\end{remark}
The following example shows that the upper bound in Proposition \ref{pr:upper_bound} is tight.
\begin{example}
Let $X_1$ and $X_2$ be two independent r.v.'s, with $X_1$ and $X_2$ being uniformly distributed over $[M]$ and $[kM]$, respectively, for some positive integers $k,M$. Set $W\triangleq X_1+X_2\ \textnormal{mod}\ M$. It can be readily verified that $W\independent X_i,\ i\in[2]$, and hence, the upper bound in Proposition \ref{pr:upper_bound} reduces to $I_s\leq H(W)$, since $H(W|X_1,X_2)=0$. 
Setting $Y\triangleq X_1+X_2\ \textnormal{mod}\ M$ attains this bound, and we also have $W-(X_1,X_2)-Y$ form a Markov chain\footnote{In this example, instead of writing $Y\triangleq W$, we used $Y\triangleq X_1+X_2\ \textnormal{mod}\ M$ to emphasize on the fact that the privacy mapping does not have access to $W$ in general.}.
\end{example}


\section{The optimal synergistic disclosure mapping}
\label{sec:3}

This section presents a practical method for computing the optimal latent feature disclosure strategy/mapping under perfect sample privacy. In what follows, we assume that $\mbox{nul}(\mathbf{P})\neq 0$, since otherwise we have from Proposition \ref{Prop1} that $I_s=0$, making the result trivial.

\subsection{General solution}

Before stating the main result of this section in Theorem \ref{theo1}, some essential preliminaries are needed as follows. Consider the singular value decomposition (SVD) of $\mathbf{P}$, which gives $\mathbf{P}=\mathbf{U}\mathbf{\Sigma}\mathbf{V}^T$ with the matrix of right eigenvectors being
\begin{equation}\label{col}
\mathbf{V}=\begin{bmatrix}
\mathbf{v}_1&\mathbf{v}_2&\dots&\mathbf{v}_{|\hat{\mathcal{X}}|}
\end{bmatrix}_{|\hat{\mathcal{X}}|\times|\hat{\mathcal{X}}|}.
\end{equation}
By assuming (without loss of generality) that the singular values are arranged in a descending order, only the first $\text{rank}(\mathbf{P})$ singular values are non-zero. Therefore, it is direct to check that the null space of $\mathbf{P}$ is given by
\begin{equation}\label{nulls}
\mbox{Null}(\mathbf{P})=\mbox{Span}\{\mathbf{v}_{\mbox{rank}(\mathbf{P})+1},\mathbf{v}_{\mbox{rank}(\mathbf{P})+2},\ldots,\mathbf{v}_{|\hat{\mathcal{X}}|}\}.
\end{equation}
Let $\mathbf{A}\triangleq\begin{bmatrix}\mathbf{v}_1&\mathbf{v}_2&\dots&\mathbf{v}_{\mbox{rank}(\mathbf{P})}\end{bmatrix}^T$ which, due to the orthogonality of the columns of $\mathbf{V}$, has the useful property $\text{Null}(\mathbf{P}) = \text{Null}(\mathbf{A})$. From Lemma~\ref{lemma2}, having $X_i\independent Y,\ \forall i\in[n]$ is equivalent to
\begin{equation}\label{nul}
\mathbf{A}(\mathbf{p}_{X^n}-\mathbf{p}_{X^n|y})=\mathbf{0},\ \forall y\in\mathcal{Y}.
\end{equation}

Let $\mathbb{S}$ be defined as
\begin{equation}\label{poly}
\mathbb{S}\triangleq\bigg\{\mathbf{t}\in\mathbb{R}^{|\hat{\mathcal{X}}|}\bigg|\mathbf{A}\mathbf{t}=\mathbf{A}\mathbf{p}_{X^n}\ ,\ \mathbf{t}\geq 0\bigg\},
\end{equation}
which is a convex polytope in $\mathcal{P}(\hat{\mathcal{X}})$, since it can be written as the intersection of a finite number of half-spaces in $\mathcal{P}(\hat{\mathcal{X}})$.

In the Markov chain $W-X^n-Y$ with $p_{Y|X^n}\in \mathcal{A}$, one can see from \eqref{nul} that $\mathbf{p}_{X^n|y}\in\mathbb{S},\ \forall y\in\mathcal{Y}$. On the other hand, for any $p_{X^n,Y}$ for which $\mathbf{p}_{X^n|y}\in\mathbb{S},\ \forall y\in\mathcal{Y}$, it is guaranteed that if one uses the corresponding mapping $p_{Y|X^n}$ to build a Markov chain $W-X^n-Y$, then the condition $X_i\independent Y,\forall i \in[n]$ holds. The above arguments prove the following equivalence:
\begin{equation}\label{siz}
W-X^n-Y,\ p_{Y|X^n}\in\mathcal{A}  \Longleftrightarrow \mathbf{p}_{X^n|y}\in\mathbb{S},\ \forall y\in\mathcal{Y}.
\end{equation}
\begin{proposition}
The supremum in (\ref{def}) is attained, and hence, it is a maximum. Furthermore, it is sufficient to have $|\mathcal{Y}|\leq \textnormal{nul}(\mathbf{P})+1$. 
\end{proposition}
\begin{proof}
The proof is provided in Appendix \ref{app1}. Later, in Corollary \ref{cor2}, it is shown that it is necessary to have $|\mathcal{Y}|\geq \left \lceil \frac{|\hat{\mathcal{X}}|}{\mbox{rank}(\mathbf{P})}\right\rceil$. 
\end{proof}
\begin{theorem}\label{theo1}
The maximizer in (\ref{def}), i.e., the optimal mapping $p^*_{Y|X^n}$, can be obtained as the solution to a standard LP.
\end{theorem}

\begin{proof}
We can express \eqref{def} as
\begin{align}
I_\text{s}=&\;H(W) - \min_{\substack{p_{Y|X^n}\in\mathcal{A}:\\W-X^n-Y}} H(W|Y)\\
=&\;H(W)-\min_{
\substack{p_Y(\cdot), \mathbf{p}_{X^n|y}\in\mathbb{S},\ \forall y\in\mathcal{Y}: \\  \sum_y p_Y(y)\mathbf{p}_{X^n|y}=\mathbf{p}_{X^n}}
}
\sum_{y} p_Y(y)H\left(\mathbf{P}_{W|X^n}\mathbf{p}_{X^n|y}\right),
\label{min}
\end{align}
where, since the minimization is over $\mathbf{p}_{X^n|y}$ rather than $p_{Y|X^n}$, the constraint $\sum_y p_Y(y)\mathbf{p}_{X^n|y}=\mathbf{p}_{X^n}$ has been added to preserve the distribution $\mathbf{p}_{X^n}$ specified in $p_{W,X^n}$.

\begin{lemma}\label{P2}
When minimizing $H(W|Y)$ over $\mathbf{p}_{X^n|y}\in\mathbb{S}$ in (\ref{min}), it is sufficient to consider only the extreme points of $\mathbb{S}$. 
\end{lemma}
\begin{proof}
Let $\mathbf{p}$ be an arbitrary point in $\mathbb{S}$. Note that the set $\mathbb{S}$ is an bounded $d$-dimensional convex subset of $\mathbb{R}^{|\hat{\mathcal{X}}|}$, where $d\leq (|\hat{\mathcal{X}}|-1)$. Therefore, any point in $\mathbb{S}$ can be written as a convex combination of at most $|\hat{\mathcal{X}}|$ extreme points of $\mathbb{S}$. Hence, $\mathbf{p}$ can be written as $\mathbf{p}=\sum_{i=1}^{|\hat{\mathcal{X}}|}\alpha_i\mathbf{p}_i$,
where $\alpha_i\geq 0\, (\forall i\in[|\hat{\mathcal{X}}|]),$ with $\sum_{i=1}^{|\hat{\mathcal{X}}|}\alpha_i=1$, and $\mathbf{p}_i \ (\forall i\in[|\hat{\mathcal{X}}|])$ are the extreme points of $\mathbb{S}$ with $\mathbf{p}_i\neq\mathbf{p}_j$ ($i\neq j$).
Due to the concavity of the entropy, one has that
\begin{equation}\label{strict}
H(\mathbf{P}_{W|X^n}\mathbf{p})\geq\sum_{i=1}^{|\hat{\mathcal{X}}|}\alpha_iH(\mathbf{P}_{W|X^n}\mathbf{p}_i).
\end{equation}
Therefore, from (\ref{strict}), it is sufficient to consider only the extreme points of $\mathbb{S}$ in the minimization.
\end{proof}

Using Lemma \ref{P2}, the optimization in (\ref{min}) can be solved in two steps: a first step in which the extreme points of set $\mathbb{S}$ are identified, followed by a second step where proper weights over these extreme points are obtained to minimize the objective function.

For the first step, we first note that the extreme points of $\mathbb{S}$ are the corresponding \textit{basic feasible solutions} (c.f. \cite{LP1}, \cite{LP2}) of the polytope in standard form
\begin{equation*}
\bigg\{\mathbf{t}\in\mathbb{R}^{|\hat{\mathcal{X}}|}\bigg|\mathbf{A}\mathbf{t}=\mathbf{b}\ ,\ \mathbf{t}\geq 0\bigg\},
\end{equation*}
with $\mathbf{b}=\mathbf{A}\mathbf{p}_{X^n}$. A standard procedure to find the extreme points of $\mathbb{S}$ is as follows~\cite[Sec. 2.3]{LP1}. Pick a set $\mathcal{B}\subset[|\hat{\mathcal{X}}|]$ of indices that correspond to $\mbox{rank}(\mathbf{P})$ linearly independent columns of matrix $\mathbf{A}$. Let $\mathbf{A}_{\mathcal{B}}$ be a $\mbox{rank}(\mathbf{P})\times\mbox{rank}(\mathbf{P})$ matrix whose columns are the columns of $\mathbf{A}$ indexed by the indices in $\mathcal{B}$. Also, for any $\mathbf x\in\mathbb{S}$, let $\tilde{\mathbf x}=\begin{bmatrix}\mathbf{x}_{\mathcal{B}}^T&\mathbf{x}_{\mathcal{N}}^T\end{bmatrix}^T$, where $\mathbf{x}_{\mathcal{B}}$ and $\mathbf{x}_{\mathcal{N}}$ are $\mbox{rank}(\mathbf{P})$-dimensional and $\mbox{nul}(\mathbf{P})$-dimensional vectors whose elements are the elements of $\mathbf{x}$ indexed by the indices in $\mathcal{B}$ and $[|\mathcal{X}|]\backslash\mathcal{B}$, respectively. For any basic feasible solution $\mathbf{x}^*$, there exists a set $\mathcal{B}\subset[|\mathcal{X}|]$ of indices that correspond to a set of linearly independent columns of $\mathbf{A}$, such that the corresponding vector of $\mathbf{x}^*$, i.e. $\tilde{\mathbf{x}}^*=\begin{bmatrix}{\mathbf{x}^*_\mathcal{B}}^T&{\mathbf{x}^*_\mathcal{N}}^T\end{bmatrix}^T$, satisfies the following
\begin{equation*}
\mathbf{x}_\mathcal{N}^*=\mathbf{0},\ \ \ \mathbf{x}_\mathcal{B}^*=\mathbf{A}_\mathcal{B}^{-1}\mathbf{b},\ \ \ \mathbf{x}_\mathcal{B}^*\geq 0.
\end{equation*}
On the other hand, for any set $\mathcal{B}\subset[|\mathcal{X}|]$ of indices that correspond to a set of linearly independent columns of $\mathbf{A}$, 
if $\mathbf{A}_\mathcal{B}^{-1}\mathbf{b}\geq 0$, then $\begin{bmatrix}
\mathbf{A}_\mathcal{B}^{-1}\mathbf{b}\\\mathbf{0}_\mathcal{N}
\end{bmatrix}$ is the corresponding vector of a basic feasible solution.
Hence, the extreme points of $\mathbb{S}$ are obtained as mentioned above, and their number is upper bounded by $|\hat{\mathcal{X}}|\choose {\mbox{rank}(\mathbf{P})}$. The general procedure of finding the extreme points of $\mathbb{S}$ is shown as a pseudocode in Algorithm~\ref{alg1}.
\begin{algorithm}\label{alg1}
\caption{Finding the extreme points of $\mathbb{S}$}
\begin{algorithmic}[1]
\Function{FindExtremePoints}{$\mathbf{p}_{X^n}$}
\State $\mathbf{P}$ =  \texttt{BuildMatrix}$(\mathbf{p}_{X^n})$
\State $\mathbf{U},\mathbf{\Sigma},[\mathbf{v}_1,\dots,\mathbf{v}_{|\hat{\mathcal{X}}|}]^T = \text{SVD}(\mathbf{P})$
\State $\mathbf{A} = [\mathbf{v}_1,\dots,\mathbf{v}_{\text{rank}(\mathbf{\Sigma})}]^T$
\State $\mathbf{b} = \mathbf{A}\mathbf{p}_{X^n}$
\State $K=0$
\State $\mathcal{B}_1,\dots,\mathcal{B}_J$ = Subsets of $[|\hat{\mathcal{X}}|]$ with cardinality $\text{rank}(\mathbf{\Sigma})$  
\For{$j=1,\dots,J$}
\If{$\mathbf{A}^{-1}_{\mathcal{B}_j} \mathbf{b} \geq 0$}
\State $K=K+1$
\State $\mathbf{p}_K = [ \mathbf{A}^{-1}_{\mathcal{B}_j} \mathbf{b}, \mathbf{0}_\mathcal{N} ]^T$
\EndIf
\EndFor
\State \Return $\mathbf{p}_1,\dots,\mathbf{p}_K$
\EndFunction
\end{algorithmic}
\end{algorithm}

For the second step, assume that the extreme points of $\mathbb{S}$, found in the first step, are denoted by $\mathbf{p}_1,\mathbf{p}_2,\ldots,\mathbf{p}_K$. Then, computing (\ref{min}) is equivalent to solving
\begin{align}
H(W)-&\min_{\mathbf{u}\geq 0}\ \begin{bmatrix}
H(\mathbf{P}_{W|X^n}\mathbf{p}_1)&\dots&H(\mathbf{P}_{W|X^n}\mathbf{p}_K)
\end{bmatrix}\cdot\mathbf{u}\nonumber\\
&\ \ \ \mbox{s.t. }\begin{bmatrix}
\mathbf{p}_1&\mathbf{p}_2&\dots&\mathbf{p}_K
\end{bmatrix}\mathbf{u}=\mathbf{p}_{X^n},\label{LP}
\end{align}
where $\mathbf{u}$ is a $K$-dimensional weight vector, and it can be verified that the constraint $\mathbf{1}_K^T\cdot \mathbf{u}=1$ is satisfied if the constraint in (\ref{LP}) is met. The problem in (\ref{LP}) is a standard LP.
\end{proof}
\begin{corollary}\label{cor2}
In the evaluation of (\ref{def}), it is necessary to have $|\mathcal{Y}|\geq \left \lceil \frac{|\hat{\mathcal{X}}|}{\textnormal{rank}(\mathbf{P})}\right\rceil$.
\end{corollary}
\begin{proof}
From the procedure of finding the extreme points of $\mathbb{S}$, it is observed that these points have at most $\textnormal{rank}(\mathbf{P})$ non-zero elements. Therefore, in order to write the $|\hat{\mathcal{X}}|$-dimensional probability vector $\mathbf{p}_{X^n}$ as a convex combination of the extreme points of $\mathbb{S}$, at least $\left\lceil \frac{|\hat{\mathcal{X}}|}{\mbox{rank}(\mathbf{P})}\right\rceil$ points are needed, which results in $|\mathcal{Y}|\geq\left\lceil \frac{|\hat{\mathcal{X}}|}{\mbox{rank}(\mathbf{P})}\right\rceil$.
\end{proof}
\begin{corollary}\label{cor1}
For a given dataset $X^n$, we can write
\begin{align}
\min_{\substack{p_{Y|X^n}\in\mathcal{A}}}H(X^n|Y)&\leq\log(\textnormal{rank}(\mathbf{P}))\label{uni1}\\
&\leq \log\bigg(\min\bigg\{\sum_{i=1}^n|\mathcal{X}_i|-n+1,|\hat{\mathcal{X}}|\bigg\}\bigg)\label{uni2}.
\end{align}
\begin{proof}
In (\ref{uni1}), we have used the fact that i) it is sufficient to consider those $\mathbf{p}_{X^n|y}$ that belong to the set of extreme points of $\mathbb{S}$, ii) these extreme points have at most $\textnormal{rank}(\mathbf{P})$ non-zero elements, and iii) entropy is maximized by the uniform distribution. The upper bound in (\ref{uni2}) follows from the fact that the rows of $\mathbf{P}$ are linearly dependent, since we have $\mathbf{1}_{|\mathcal{X}_i|}^T\cdot\mathbf{P}_{X_i|X^n}=\mathbf{1}_{|\hat{\mathcal{X}}|}^T,\ \forall i\in[n]$, which means that there are at most $\sum_{i=1}^n|\mathcal{X}_i|-(n-1)$ linearly independent rows in $\mathbf{P}$.
\end{proof}
 
\end{corollary}

Following the proof of Theorem~\ref{theo1}, Algorithm~\ref{alg2} provides a summary of how to compute the optimal disclosure mapping, using as inputs $p_{W,X^n}$. Its procedure is illustrated in example~\ref{Ex1}. Although it serves it purpose, the performance of Algorithm~\ref{alg2} scales poorly with the dataset size $n$. Suboptimal procedures to build perfectly-private mappings are discussed in Section~\ref{sec:heu}. 
\begin{algorithm}\label{alg2}
\caption{Building the optimal disclosure mapping $p^*_{Y|X^n}$}
\begin{algorithmic}[1]
\Function{FindOptimalMapping}{$\mathbf{P}_{W|X^n},\mathbf{p}_{X^n}$}
\State $\mathbf{p}_1,\dots, \mathbf{p}_K$ = \texttt{FindExtremePoints}$(\mathbf{p}_{X^n})$
\For{$k=1,\dots,K$}
\State $c_k = H(\mathbf{P}_{W|X^n}\mathbf{p}_k)$
\EndFor
\State Find $\mathbf{u}^* = \text{Argmin} \sum_{k=1}^K u_k c_k$ s.t. $[\mathbf{p}_1,\dots,\mathbf{p}_K]\mathbf{u} = \mathbf{p}_{X^n}$ and $\mathbf{u}\geq 0$
\State $L=0$
\For{$k=1,\dots,K$}
\If{$u_k > 0$}
\State $L=L+1$
\State $p(Y=L) = u_k$
\State $\mathbf{p}_{X^n|Y=L} = \mathbf{p}_k$
\EndIf
\EndFor
\State $\mathbf{p}_Y = [p(Y=1),\dots,p(Y=L)]$
\State $\mathbf{P}_{X^n|Y} = [\mathbf{p}_{X^n|Y=1},\dots,\mathbf{p}_{X^n|Y=L}]$
\State $\mathbf{P}_{Y|X^n} = \text{diag}(\mathbf{p}_Y) \cdot \mathbf{P}_{X^n|Y}^T \cdot \text{diag}(\mathbf{p}_{X^n})^{-1}$
\State \Return $\mathbf{P}_{Y|X^n}$
\EndFunction
\end{algorithmic}
\end{algorithm}
\begin{example}\label{Ex1}
Let $W\sim\mbox{Bern}(\frac{1}{2})$ be the random variable that the user wishes to share with an analyst, and assume that the user has data samples denoted by $X_1$ and $X_2$, which are, respectively, the observations of $W$ through a binary symmetric channel with crossover probability $\alpha$, i.e., BSC($\alpha$), and a binary erasure channel with erasure probability $e$, i.e., BEC($e$). Figure \ref{F1} provides an illustrative representation of this setting. Set $\alpha=\frac{2}{3}$, and $e = \frac{1}{2}$, which results in $\mathbf{p}_{X^2}=\frac{1}{12}\begin{bmatrix}
1&3&2&2&3&1
\end{bmatrix}^T$, and 
\begin{equation}
\mathbf{P}_{W|X^2}=\begin{bmatrix}
1&\frac{2}{3}&0&1&\frac{1}{3}&0\\
0&\frac{1}{3}&1&0&\frac{2}{3}&1
\end{bmatrix}.
\end{equation}
Matrix $\mathbf{P}$ in (\ref{matp}) is given by
\begin{equation*}
\mathbf{P}=\begin{bmatrix}
1&1&1&0&0&0\\0&0&0&1&1&1\\1&0&0&1&0&0\\0&1&0&0&1&0\\0&0&1&0&0&1
\end{bmatrix},
\end{equation*}
and by obtaining an SVD of $\mathbf{P}$, we obtain matrix $\mathbf{A}$\footnote{Note that $\mathbf{A}$ is not unique.} as 
\begin{equation*}
\footnotesize
\mathbf{A}=
\frac{1}{10^{4}}\begin{bmatrix}
4082&4082&4082&4082&4082&4082\\
4082&4082&4082&-4082&-4082&-4082\\
-4677&5270&-593&-4677&5270&-593\\
-3385&-2385&5743&-3385&-2358&5743
\end{bmatrix}.
\end{equation*}
There are at most 15 ways of choosing 4 linearly independent columns of $\mathbf{A}$. From $\mathbf{x}_{\mathcal{B}}=\mathbf{A}_{\mathcal{B}}^{-1}\mathbf{A}\mathbf{p}_{X^2}$, and the condition $\mathbf{x}_{\mathcal{B}}\geq 0$, we obtain the extreme points of $\mathbb{S}$ as
\begin{align*}
\mathbf{p}_{1}=\begin{bmatrix}
\frac{1}{4}\\0\\\frac{1}{4}\\0\\\frac{1}{2}\\0
\end{bmatrix}, 
\mathbf{p}_{2}=\begin{bmatrix}
0\\ \frac{1}{2}\\0\\\frac{1}{4}\\0\\\frac{1}{4}
\end{bmatrix},
\mathbf{p}_{3}=\begin{bmatrix}
\frac{1}{4}\\\frac{1}{4}\\0\\0\\\frac{1}{4}\\\frac{1}{4}
\end{bmatrix},
\mathbf{p}_{4}=\begin{bmatrix}
0\\\frac{1}{4}\\\frac{1}{4}\\\frac{1}{4}\\\frac{1}{4}\\0
\end{bmatrix}.
\end{align*}
Finally, the LP is given by
\begin{align}
&\min_{\mathbf{u}\geq 0}\ \begin{bmatrix}
H(\mathbf{P}_{W|X^2}\mathbf{p}_1)&\ldots&H(\mathbf{P}_{W|X^2}\mathbf{p}_4)
\end{bmatrix}\cdot\mathbf{u}=0.9866\mbox{ bits}\nonumber\\
&\ \ \ \mbox{s.t. }\begin{bmatrix}
\mathbf{p}_1&\mathbf{p}_2&\mathbf{p}_3&\mathbf{p}_4
\end{bmatrix}\mathbf{u}=\mathbf{p}_{X^2},
\end{align}
where $\mathbf{u}^*=\begin{bmatrix}
\frac{1}{3}&\frac{1}{3}&0&\frac{1}{3}
\end{bmatrix}^T$. Therefore, the maximum information that can be shared with an analyst about $W$, while preserving the privacy of the observations, is $I_s=0.0134$ bits, which is achieved by the following synergistic disclosure strategy
\begin{equation}
\mathbf{P}^*_{Y|X^2}=\begin{bmatrix}
1&0&\frac{1}{2}&0&\frac{2}{3}&0\\
0&\frac{2}{3}&0&\frac{1}{2}&0&1\\
0&\frac{1}{3}&\frac{1}{2}&\frac{1}{2}&\frac{1}{3}&0
\end{bmatrix}.
\end{equation}
\end{example}

    \begin{figure}[!t]
    \begin{center}
    \begin{tikzpicture}
      \node[dspnodeopen, minimum width=4pt, dsp/label=above] (X1_1) {\large{$X_1$}};    
      \node[dspnodeopen, minimum width=4pt, below=1.5cm of X1_1, dsp/label=above] (X1_0) {};
      \node[dspnodeopen, minimum width=4pt, right=2.4cm of X1_1, dsp/label=above] (W_1) {\large{$W$}};    
      \node[dspnodeopen, minimum width=4pt, below=1.5cm of W_1, dsp/label=above] (W_0) {};
      \node[dspnodeopen, minimum width=4pt, below right=0.7cm and 2.4cm of W_1, dsp/label=above] (X2_0) {};
      \node[dspnodeopen, minimum width=4pt, above=1.2cm of X2_0, dsp/label=above] (X2_1) {\large{$X_2$}};
      \node[dspnodeopen, minimum width=4pt, below=1.2cm of X2_0, dsp/label=above] (X2_a) {};
    \draw[line width=1.2pt] (X1_1) to node[midway, above]{\small{$1-\alpha$}} (W_1);
    \draw[line width=1.2pt] (X1_1) to node[near end, above]{\small{$\alpha$}} (W_0);
    \draw[line width=1.2pt] (X1_0) to node[near end, below]{\small{$\alpha$}} (W_1);
    \draw[line width=1.2pt] (X1_0) to node[midway, below]{\small{$1-\alpha$}} (W_0);
    \draw[line width=1.2pt] (W_1) to node[midway, above]{\small{$1-e$}} (X2_1);
    \draw[line width=1.2pt] (W_1) to node[near start, below]{\small{$e$}} (X2_0);
    \draw[line width=1.2pt] (W_0) to node[near start, above]{\small{$e$}} (X2_0);
    \draw[line width=1.2pt] (W_0) to node[midway, below]{\small{$1-e$}} (X2_a);
    \end{tikzpicture}
    \caption{Example \ref{Ex1}, where $X_1$ and $X_2$ are observations of $W$ through a $BSC(\alpha)$ and a $BEC(e)$, respectively.} \label{F1}
    \end{center}
    \end{figure}
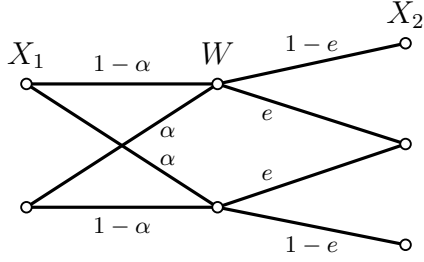


\subsection{Two binary samples}
\label{tworand}

To illustrate the above results, in what follows, we consider the case where two binary (noisy) observations $X_1,X_2$ of an underlying phenomenon $W$ are available. As before, the goal is to maximally inform an analyst about $W$, while preserving the privacy of both observations.

Consider the tuple $(W,X_1,X_2)$ distributed according to a given joint distribution 
$p_{W,X_1,X_2} = p_{X_1,X_2} p_{W|X_1,X_2} $. In this setting, no condition is imposed on the conditional $p_{W|X_1,X_2}$. Without loss of generality, $p_{X_1,X_2}$ is parametrized as
\begin{equation}\label{px}
\mathbf{p}_{X^2}=\begin{bmatrix}\alpha-r&r&(\beta-\alpha)+r&(1-\beta)-r \end{bmatrix}^T,
\end{equation}
where $\alpha,\beta \in (0,1)$ are degrees of freedom that determine the marginals , i.e., $X_1\sim\mbox{Bern}(\alpha)$ and $X_2\sim\mbox{Bern}(\beta)$, while $r\in[0,R]$ with $R\triangleq \min\{\alpha,1-\beta\}$ determines the interdependency between $X_1$ and $X_2$. In particular, $X_1\independent X_2$, if and only if $r=\alpha(1-\beta)$. 

If $r\in(0,R)$\footnote{For the uninteresting cases where $r\in\{0,R\}$, we have $|\hat{\mathcal{X}}|<4$ and $\mbox{nul}(\mathbf{P})=0$. Consequently, from Proposition~\ref{Prop1},  we get $I_s=0$.}, we have $\hat{\mathcal{X}}=\{(0,0),(0,1),(1,0),(1,1)\}$, and correspondingly one finds that
\begin{equation*}
\mathbf{P}=
\begin{bmatrix}
1&1&0&0\\0&0&1&1\\1&0&1&0\\0&1&0&1
\end{bmatrix}.
\end{equation*}
A direct calculation shows that $\text{Null}(\mathbf{P})$ is spanned by the single vector $\mathbf{n}=\begin{bmatrix}1&-1&-1&1\end{bmatrix}^T$. As the null space of $\mathbf{P}$ is one-dimensional, one can check that $\mathbb{S}$ has only two extreme points given by $\mathbf{a}_1 = \mathbf{p}_{X^2} - (R-r)\mathbf{n}$ and $\mathbf{a}_2 = \mathbf{p}_{X^2} +r\mathbf{n}$ (see Figure~\ref{F2}). Note that the original distribution can be recovered as a convex combination of these two extreme points, i.e.,
\begin{equation}\label{recover}
\mathbf{p}_{X^2} = \frac{r}{R}\mathbf{a}_1 + \frac{R-r}{R}\mathbf{a}_2.
\end{equation}
Therefore, using \eqref{min}, $I_\text{s}$ can be computed as
\begin{align}
I_\text{s} &= H(W) - \frac{r}{R}H(\mathbf{P}_{W|X^2}\mathbf{a}_1) 
- \frac{R-r}{R}H(\mathbf{P}_{W|X^2}\mathbf{a}_2)\nonumber\\
&=H(\mathbf{p}_W)- \frac{r}{R}H\left(\mathbf{p}_W-(R-r)\mathbf{P}_{W|X^2}\mathbf{n}\right)\nonumber\\
&\ \ \ -\frac{R-r}{R}H\left(\mathbf{p}_W+r\mathbf{P}_{W|X^2}\mathbf{n}\right). \label{sol}
\end{align}
From the last expression, it is direct to verify that, $I_\text{s} >0$ if and only if $\mathbf{n}\not\in\mbox{Null}(\mathbf{P}_{W|X^2})$. 
\begin{figure}
\begin{centering}
\includegraphics[scale=0.275]{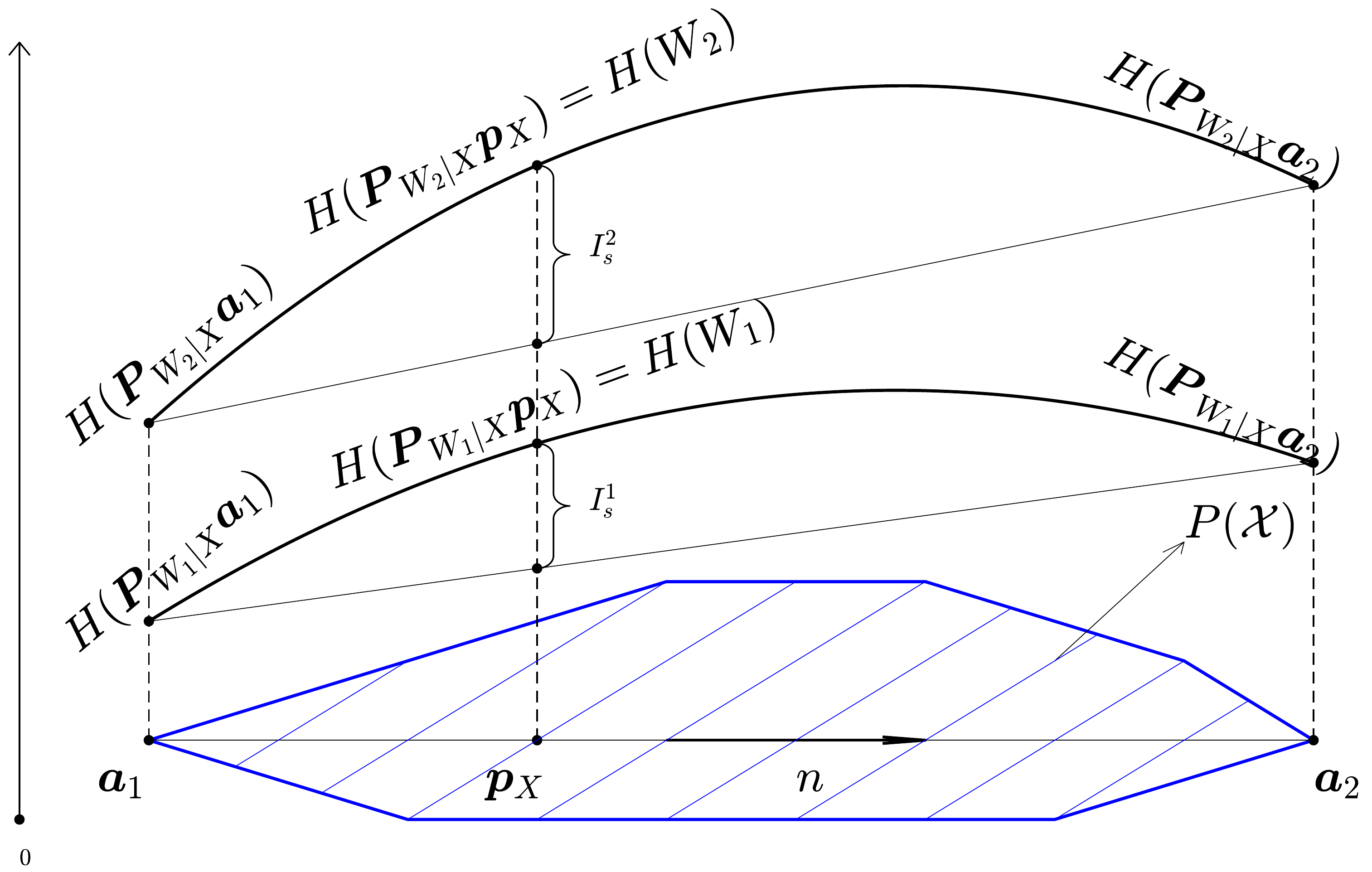}
\caption{Diagram of private information disclosure for two tuples $(W_1,X_1,X_2)$ and $(W_2,X_1,X_2)$, where $(X_1,X_2)$ are binary and distributed according to $\mathbf{p}_X$ as given in \eqref{px}, and $p_{W_1|X}\neq p_{W_1|X}$. While their private disclosure capacities, i.e., $I_s^i,\ i=1,2$, are different, their optimal synergistic disclosure strategies are the same, as regardless of the tuples, we have $\mathbf{p}_X = \frac{r}{R}\mathbf{a}_1 + \frac{R-r}{R}\mathbf{a}_2.$}\label{F2}
\par\end{centering}
\vspace{0mm}
\end{figure}
Finally, the optimal mapping $\mathbf{P}^*_{Y|X^2}$ is derived as follows. Considering \eqref{recover}, let $\mathcal{Y}\triangleq\{y_1,y_2\}$, and fix $p_Y(y_1)=\frac{r}{R}$ and $\mathbf{p}_{X^2|y_i}=\mathbf{a}_i, i=1,2$. Using these, a direct calculation results in the following optimal mapping 
\begin{equation}\label{ma}
\mathbf{P}^*_{Y|X^2}=\begin{bmatrix}\frac{r(\alpha-R)}{R(\alpha-r)}&1&\frac{r(\beta-\alpha+R)}{R(\beta-\alpha+r)}&\frac{r(1-\beta-R)}{R(1-\beta-r)}\\\frac{\alpha(R-r)}{R(\alpha-r)}&0&\frac{(\beta-\alpha)(R-r)}{R(\beta-\alpha+r)}&\frac{(1-\beta)(R-r)}{R(1-\beta-r)}\end{bmatrix}.
\end{equation}
It is important to note that, although the disclosure
capacity in (\ref{sol}) depends on the choice of $\mathbf{P}_{W|X^2}$,
the optimal synergistic disclosure strategy in (\ref{ma}) is only a functional of $\mathbf{p}_{X^2}$ (or equivalently, $\alpha,\beta,r$), and does not depend on $\mathbf{P}_{W|X^2}$. This observation is a special case of the following proposition.

\begin{proposition}\label{P4}
For the tuple $(W,X^n)$, in which $|\hat{\mathcal{X}}|=\sum_{i=1}^n|{\mathcal{X}_i}|-n +2$, the optimal synergistic disclosure strategy, i.e., $\mathbf{P}^*_{Y|X^n}$, does not depend on $p_{W|X^n}$.
\end{proposition}
\begin{proof}
This follows from the fact that in this setting $\textnormal{nul}(\mathbf{P})=1$, and as a result, $\mathbb{S}$ has only two extreme points\footnote{In this case, the optimal mapping conveys at most one bit about $W$, since we have $|\mathcal{Y}|=2$.}. Therefore, the mere condition of preserving $\mathbf{p}_{X^n}$ suffices to define the probability masses of these two extreme points. Hence, the LP is solved by its constraint, not being affected by the choice of $W$. 
\end{proof}

This result implies that the same strategy/mapping can provide an optimal service in addressing any possible query over the data, as given by a specific $p_{W|X^n}$. In other words, optimal processing of the data can be done in the absence of any knowledge about the query. However, this does not hold in general.

\section{Asymptotic performance on large datasets}\label{sec:large_datasets}

In this section we analyse datasets that are composed by noisy measurements $X_1,\dots,X_n$ of a variable of interest $W$, and focus on their asymptotic properties. For the sake of tractability, we focus in the case where the noise that affects each measurement is conditionally independent and identically distributed given $W$. In the sequel, Section~\ref{sec:prel2} introduces tools that are later used in our analysis, which is outlined in Section~\ref{sec:IIDobs}.

\subsection{Preliminaries}\label{sec:prel2}
For a pair of random variables $(X,W)\in \mathcal X\times\mathcal W$, with finite alphabets, 
following \cite{Shahab1}, we define\footnote{In \cite{Shahab1}, the authors name $C_X(W)$ \textit{private information about $X$ carried by $W$}.}
\begin{equation}\label{pr}
C_X(W)\triangleq \min_{\substack{U:X-U-W,\\H(U|W)=0}}H(U).
\end{equation}
Since $H(U|W)=0$ implies that $U$ is a deterministic function of $W$, (\ref{pr}) means that among all the functions of $W$ that make $X$ and $W$ conditionally independent, we want to find the one with the lowest entropy. 
It can be verified that 
\begin{equation*}
I(X;W)\leq C_X(W)\leq H(W),
\end{equation*}
where the first inequality is due to the data processing inequality applied on the Markov chain $X-U-W$, i.e., $I(U;W)\geq I(X;W)$, and the second inequality is a direct result of the fact that $U=W$ satisfies the constraints in (\ref{pr}).

Let $T^{\mathcal{X}}:\mathcal{W}\to\mathcal{P}(\mathcal{X})$ be a mapping from $\mathcal{W}$ to the probability simplex on $\mathcal{X}$ defined by $w\to p_{X|W}(\cdot|w)$. It is shown in \cite[Theorem 3]{Shahab1} that the minimizer in (\ref{pr}) is $U^*=T^{\mathcal{X}}(W)$; furthermore, it is proved in \cite[Lemma 5]{Shahab1} that $C_X(W)=H(W)$ if and only if there do not exist $w_1,w_2\in\mathcal{W}$ such that $p_{X|W}(\cdot|w_1)=p_{X|W}(\cdot|w_2)$. In the sequel, for a given pmf $p_{W,X}$, we denote $U^*$ by $\tilde{W}$, and hence, we have $H(\tilde{W})=C_X(W)$. Moreover, $W-\tilde{W}-X$ and $\tilde{W}-W-X$ are Markov chains. Figure \ref{fig:w_tilde} provides an example of $\tilde{W}$ for a given $p_{W,X}$.
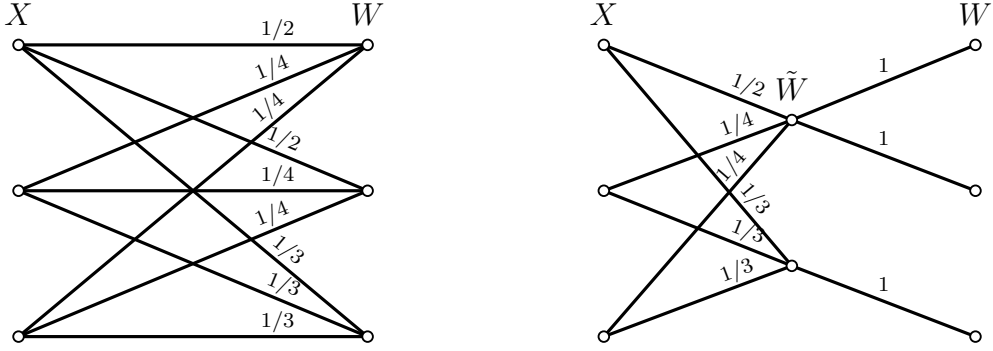
\begin{figure}[!t]\label{fig:w_tilde}
\begin{center}
\begin{tikzpicture}
  \node[dspnodeopen, minimum width=4pt, dsp/label=above] (X1_1) {\large{$X$}};    
  \node[dspnodeopen, minimum width=4pt, below=1.8cm of X1_1, dsp/label=above] (X1_0) {};
  \node[dspnodeopen, minimum width=4pt, below=1.8cm of X1_0, dsp/label=above] (X1_m1) {};
  \node[dspnodeopen, minimum width=4pt, right=4.5cm of X1_1, dsp/label=above] (W_1) {\large{$W$}};    
  \node[dspnodeopen, minimum width=4pt, below=1.8cm of W_1, dsp/label=above] (W_0) {};
  \node[dspnodeopen, minimum width=4pt, below=1.8cm of W_0, dsp/label=above] (W_m1) {};
  \node[dspnodeopen, minimum width=4pt, right=3 of W_1, dsp/label=above] (X2_1) {\large{$X$}};    
  \node[dspnodeopen, minimum width=4pt, below=1.8cm of X2_1, dsp/label=above] (X2_0) {};
  \node[dspnodeopen, minimum width=4pt, below=1.8cm of X2_0, dsp/label=above] (X2_m1) {};
  \node[dspnodeopen, minimum width=4pt, below right=0.9cm and 2.4cm of X2_1, dsp/label=above] (W1_1) {\large{$\tilde{W}$}};
  \node[dspnodeopen, minimum width=4pt, below=1.8cm of W1_1, dsp/label=above] (W1_0) {};
  \node[dspnodeopen, minimum width=4pt, right=4.8cm of X2_1, dsp/label=above] (W2_1) {\large{$W$}};    
  \node[dspnodeopen, minimum width=4pt, below=1.8cm of W2_1, dsp/label=above] (W2_0) {};
  \node[dspnodeopen, minimum width=4pt, below=1.8cm of W2_0, dsp/label=above] (W2_m1) {};
\draw[line width=1.2pt] (X1_1) to node[near end, inner sep=1pt, above]{\footnotesize{$1/2$}} (W_1);
\draw[line width=1.2pt] (X1_0) to node[near end, inner sep=1pt, above, sloped]{\footnotesize{$1/4$}} (W_1);
\draw[line width=1.2pt] (X1_m1) to node[near end, inner sep=1pt, above, sloped]{\footnotesize{$1/4$}} (W_1);
\draw[line width=1.2pt] (X1_1) to node[near end, above, inner sep=1pt, sloped]{\footnotesize{$1/2$}} (W_0);
\draw[line width=1.2pt] (X1_0) to node[near end, above, inner sep=1pt, sloped]{\footnotesize{$1/4$}} (W_0);
\draw[line width=1.2pt] (X1_m1) to node[near end, above, inner sep=1pt, sloped]{\footnotesize{$1/4$}} (W_0);
\draw[line width=1.2pt] (X1_1) to node[near end, above, inner sep=1pt, sloped]{\footnotesize{$1/3$}} (W_m1);
\draw[line width=1.2pt] (X1_0) to node[near end, above, inner sep=1pt, sloped]{\footnotesize{$1/3$}} (W_m1);
\draw[line width=1.2pt] (X1_m1) to node[near end, above, inner sep=1pt, sloped]{\footnotesize{$1/3$}} (W_m1);
\draw[line width=1.2pt] (X2_1) to node[near end, above, inner sep=1pt, sloped]{\footnotesize{$1/2$}} (W1_1);
\draw[line width=1.2pt] (X2_0) to node[near end, above, inner sep=1pt, sloped]{\footnotesize{$1/4$}} (W1_1);
\draw[line width=1.2pt] (X2_m1) to node[near end, above, inner sep=1pt, sloped]{\footnotesize{$1/4$}} (W1_1);
\draw[line width=1.2pt] (X2_1) to node[near end, above, inner sep=1pt, sloped]{\footnotesize{$1/3$}} (W1_0);
\draw[line width=1.2pt] (X2_0) to node[near end, above, inner sep=1pt, sloped]{\footnotesize{$1/3$}} (W1_0);
\draw[line width=1.2pt] (X2_m1) to node[near end, above, inner sep=1pt, sloped]{\footnotesize{$1/3$}} (W1_0);
\draw[line width=1.2pt] (W1_1) to node[midway, above]{\footnotesize{$1$}} (W2_1);
\draw[line width=1.2pt] (W1_1) to node[midway, above]{\footnotesize{$1$}} (W2_0);
\draw[line width=1.2pt] (W1_0) to node[midway, above]{\footnotesize{$1$}} (W2_m1);
\end{tikzpicture}
\caption{The right figure provides a graphical representation of $\tilde{W}$ for a given joint pmf $p_{W,X}$ shown on the left figure. The probabilities on theses figures correspond to the transition from right to left. Note that we have $\tilde{W}$ as a deterministic function of $W$ in the Markov chain $X-\tilde{W}-W$.} \label{fig:w_tilde}
\end{center}
\end{figure}

To conclude this subsection, we present the notion of typical sequences. Assume that $u^n$ is an $n$-sequence whose elements are drawn from an arbitrary set denoted by $\mathcal{U}$. The \textit{type} of $u^n$ is defined as
\begin{equation}\label{typedef}
\pi(u|u^n)\triangleq\frac{|\{i|u_i=u\}|}{n},\ \forall u\in\mathcal{U}.
\end{equation}
Then, for a fixed pmf $q_U(\cdot)$ on $\mathcal{U}$, and $\epsilon\in(0,1)$, define the $\epsilon$-typical set as\footnote{Here, we pick the notion $\mathcal{T}_{\epsilon}^n(q_U(\cdot))$ over $\mathcal{T}_{\epsilon}^n(U)$ as in \cite{Elgamal}, to emphasize on the generating distribution. This is useful in the sequel when considering conditional pmfs as the underlying generator. } 
\begin{equation}\label{typic}
\mathcal{T}_{\epsilon}^n(q_U(\cdot))\triangleq\bigg\{u^n\bigg|\ |\pi(u|u^n)-q_U(u)|\leq\epsilon q_U(u),\forall u\in\mathcal{U}\bigg\}. 
\end{equation}

\subsection{Asymptotic analysis}\label{sec:IIDobs}

Let $W$ be a variable of interest that is distributed according to $ p_W$ ($|\mathcal{W}|<\infty$), and consider a dataset $X^n$ where $X_i$'s are i.i.d. conditioned on $W$ according to $p_{X|W}$ ($|\mathcal{X}|<\infty$). In other words, $p_{X^n|W}(x^n|w)=\prod_{i=1}^np_{X|W}(x_i|w),\ \forall x^n\in\mathcal{X}^n,\forall w\in\mathcal{W},\ \forall n\geq 1$. In the sequel we use $X$ (without subscript or superscript) to denote a generic sample that follows $p_{X|W}$. This model corresponds to dataset of noisy observations ($X_i$'s) of an underlying phenomenon ($W$), where the observational noise is i.i.d. 

Prior to investigating the privacy-preserving data disclosure, we characterise the total information contained in the observations about the underlying phenomenon that can be disclosed when there are no privacy constraints.
\begin{theorem}\label{th4}
We have
\begin{equation}\label{PrResult}
\lim_{n\to\infty}I(W;X^n)=C_X(W),
\end{equation}
where $C_X(W)$ is defined in (\ref{pr}).
\end{theorem}
\begin{proof}
From the definition of $\tilde{W}$, it can be verified that $W-\tilde{W}-X^n$, and $\tilde{W}-W-X^n$, where $X_i$'s are also i.i.d. conditioned on $\tilde{W}$. 
For the converse, we have
\begin{align}
I(W;X^n)
&=I(\tilde{W};X^n)\nonumber\\
&\leq H(\tilde{W})\nonumber\\
&=C_X(W).\label{convers}
\end{align}
The achievability is as follows. 
We have $p_{X|\tilde{W}}(\cdot|i)\neq p_{X|\tilde{W}}(\cdot|j), \forall i,j\in\tilde{\mathcal{W}}\ (i\neq j)$, which follows from the definition of $\tilde{W}$. As a result, for a fixed $i,j\in\tilde{\mathcal{W}}$ ($i\neq j$), there exists $x^{i,j}\in\mathcal{X}$, such that $p_{X|\tilde{W}}(x^{i,j}|i)>p_{X|\tilde{W}}(x^{i,j}|j)$.
Let $B^{i,j} \triangleq \big\{x\in\mathcal{X}\big| p_{X|\tilde{W}}(x|i)>p_{X|\tilde{W}}(x|j) \big\}$ , $\forall i,j\in\tilde{\mathcal{W}}$ ($i\neq j$). Define
\begin{equation} \epsilon^{i,j}\triangleq\min_{x\in B^{i,j}}\frac{p_{X|\tilde{W}}(x|i)-p_{X|\tilde{W}}(x|j)}{p_{X|\tilde{W}}(x|i)+p_{X|\tilde{W}}(x|j)}
\qquad \forall i,j\in\tilde{\mathcal{W}} \text{ with } i\neq j.
\end{equation}
Hence, we have $\epsilon^{i,j}>0$.
Also, let $\epsilon\triangleq\min_{\substack{i,j\in\tilde{\mathcal{W}}\\i\neq j}}\epsilon^{i,j}$, which is positive, since it is the minimum over a finite set of positive elements. It can be verified that from this choice of $\epsilon$, the $\epsilon$-typical sets corresponding to the pmfs $p_{X|\tilde{W}}(\cdot|i)$ ($\forall i\in\tilde{\mathcal{W}}$) are disjoint, i.e., $\mathcal{T}_{\epsilon}^n\bigg(p_{X|\tilde{W}}(\cdot|i)\bigg)\cap\mathcal{T}_{\epsilon}^n\bigg(p_{X|\tilde{W}}(\cdot|j)\bigg)=\emptyset, \forall i,j\in\tilde{\mathcal{W}}, i\neq j.$ Let $L_n:\mathcal{X}^n\to\tilde{\mathcal{W}}\cup\{e\}$ be defined as
\begin{equation}\label{ldef}
L_n(x^n)\triangleq\left\{\begin{array}{cc}i&\textnormal{if }x^n\in\mathcal{T}_{\epsilon}^n\bigg(p_{X|\tilde{W}}(\cdot|i)\bigg), \mbox{ for some }i\in\tilde{\mathcal{W}}\\e&\mbox{o.w.}
\end{array}\right..
\end{equation}
Let $p^{(n)}_e=\mbox{Pr}\{\tilde{W}\neq L_n\}$ denote the error probability in the Markov chain $W-X^n-L_n$. We can write
\begin{align*}
\lim_{n\to\infty}p^{(n)}_e&=\lim_{n\to\infty}\sum_{i\in\tilde{\mathcal{W}}}p_{\tilde{W}}(i)\mbox{Pr}\{L_n\neq i|\tilde{W}=i\}\\
&=\lim_{n\to\infty}\sum_{i\in\tilde{\mathcal{W}}}p_{\tilde{W}}(i)\mbox{Pr}\bigg\{X^n\not\in\mathcal{T}_{\epsilon}^n\bigg(p_{X|\tilde{W}}(\cdot|i)\bigg)\bigg|\tilde{W}=i\bigg\}\\
&=0,
\end{align*}
where the last step follows from the law of large numbers (LLN), since conditioned on $\{\tilde{W}=i\}$, $X_i$'s are i.i.d. according to $p_{X|\tilde{W}}(\cdot|i)$. Therefore, from the data processing and Fano's inequality
\begin{align}
I(\tilde{W};X^n) &\geq I(\tilde{W};L_n)\\
&=H(\tilde{W})-H(\tilde{W}|L_n)\\
&\geq H(\tilde{W})-H(p^{(n)}_e)-p^{(n)}_e\log|\tilde{\mathcal{W}}|,
\end{align}
which results in $\lim_{n\to\infty}I(\tilde{W};X^n)\geq H(\tilde{W})=C_X(W)$.
\end{proof}
\begin{remark}
The proof of achievability relies on the notion of \textit{robust typicality} \cite{Elgamal}, which is used to distinguish between different conditional pmfs of the form $p_{X|\tilde{W}}(\cdot|i)$. It is important to note that the notion of weak typicality does not suffice for this purpose. In other words, assume that the $\epsilon$-typical set for a given pmf $q_U(\cdot)$ is defined as
\begin{equation}\label{weak}
\mathcal{A}_{\epsilon}^n(q_U(\cdot))\triangleq\bigg\{u^n\bigg|\ |-\frac{1}{n}\log q_{U^n}(u^n)-H(U)|\leq\epsilon\bigg\}, 
\end{equation}
where $q_{U^n}(u^n)=\prod_{i=1}^n q_U(u_i),\forall u^n\in\mathcal{U}^n$. Let us focus in the case that $X$ and $W$ are binary and the transition from $W$ to $X$ follows a BSC($\alpha$) for $\alpha\in(0,\frac{1}{2})$. In this case $\tilde{W}=W$, and from Theorem \ref{th4}, $W$ can be inferred precisely from an infinite number of i.i.d. observations $X_i$'s. However, it can be verified that nothing can be inferred about $W$ if one uses (\ref{weak}) instead of (\ref{typic}) in the proof of achievability, since $\mathcal{A}_{\epsilon}^n\bigg(p_{X|W}(\cdot|w_1)\bigg)=\mathcal{A}_{\epsilon}^n\bigg(p_{X|W}(\cdot|w_2)\bigg),\forall \epsilon>0$.  The key difference is that While using (\ref{typic}) enables us to distinguish between different conditional pmfs of the form $p_{X|\tilde{W}}(\cdot|i)$, using (\ref{weak}) aims at doing the same task only through their corresponding conditional entropies, i.e., $H(X|\tilde{W}=i)$.  
\end{remark}
\begin{remark}
Equation (\ref{PrResult}) can be used instead of (\ref{pr}) as the definition of $C_X(W)$. In other words, for a pair $(W,X)\sim p_{W,X}$ one can define
\begin{equation*}
   C_X(W)\triangleq\lim_{n\to\infty}I(W;X^n),
\end{equation*}
where $p_{X^n|W}(x^n|w)=\prod_{i=1}^np_{X|W}(x_i|w)$; afterwards, (\ref{pr}) follows.
\end{remark}

In what follows, the asymptotic behaviour of the synergistic disclosure capacity $I_\text{s}(W;X^n)$ is investigated as the number of data samples grows. To this end, we make use of the following definition.
\begin{definition}\label{pf}
For a given pmf $p_{W,X}$, and $\tilde{W}$ as defined in Section~\ref{sec:prel2}, define
\begin{align}
    C_1(\alpha)&\triangleq \max_{\substack{p_{Y|\tilde{W}}:\\X-\tilde{W}-Y\\I(X;Y)\leq \alpha}}I(\tilde{W};Y)\label{C1}\\
    C_2(\alpha)&\triangleq \max_{\substack{p_{Y|\tilde{W},X}:\\I(X;Y)\leq\alpha}}I(\tilde{W};Y),\label{C2}
\end{align}
where $\alpha\in[0,I(X;\tilde{W})]$.
\end{definition}
The above definitions capture the utility-privacy trade-off in a hypothetical scenario, in which the curator discloses information about $\tilde{W}$ while preserving the privacy of $X$. To this end, in the case of $C_1(\cdot)$, the curator has access only to $\tilde{W}$, which is similar to the \textit{output perturbation} model in \cite{Ishwar}. In the case of $C_2(\cdot)$, the curator has the extra advantage of observing $X$, which is similar to the \textit{full data observation} model in \cite{Ishwar}.
\begin{theorem}\label{th5}
The synergistic disclosure capacity, i.e., $I_s(W,X^n)$, converges as $n$ grows. Moreover, we have that
\begin{equation}
    \lim_{n\to\infty}I_s(W,X^n)=C_1(0),
\end{equation}
where $C_1(\cdot)$ is defined in \eqref{C1}.
\end{theorem}
\begin{proof}
At first, it is not clear that $I_s(W,X^n)$ converges with $n$; on the one hand, having more data samples helps conveying some information about $W$, while on the other hand, it adds to the constraints of perfect sample privacy. In fact, $I_s$ can have a non-monotonous dependency on $n$, as shown in the example provided in Table~\ref{table:1}.
 \begin{table}[h!]
  \begin{center}
    \caption{Non-monotonic $I_s$ for $W\sim$Bernoulli($1/3$) and samples generated via a BSC with crossover probability $0.1$.}
    \vspace{0.2cm}
    \label{tab:table1}
    \begin{tabular}{|c|c c c|} 
      \hline
      $n$ & $2$ & $3$ & $4$ \\ \hline
      $I_s$ & $8.34\times 10^{-3}$ & $4.88 \times 10^{-2}$ & $4.47\times 10^{-2}$ \\ \hline
    \end{tabular}\label{table:1}
  \end{center}
\end{table}

 In spite of this, one can see that as $n$ grows, a better estimate of $\tilde{W}$ becomes available at the input of the privacy mapping, i.e, $L_n$ as defined in (\ref{ldef}). Hence, as $n$ increases, one can expect that $I_s$ gets closer to
\begin{equation}\label{jdef}
    J(W,X^n)\triangleq\max_{\substack{p_{Y|\tilde{W},X^n}:\\Y\independent X_i\ \forall i\in[n]}} I(\tilde{W};Y),
\end{equation}
which is formally stated in the following Lemma\footnote{Note that from $W-\tilde{W}-Y$ and $\tilde{W}-W-Y$, we have $I(W;Y)=I(\tilde{W};Y).$}.

\begin{lemma}\label{lem2}
We have
\begin{equation}\label{asymp1}
    \lim_{n\to\infty}J(W,X^n)-I_s(W,X^n)=0.
\end{equation}
\end{lemma}
\begin{proof}
The proof is provided in Appendix \ref{app1}.
\end{proof}
At this stage, the convergence of $I_s$ can be proved as follows. It can be verified that $J$ is a non-increasing function of $n$, since by increasing $n$, the number of privacy constraints increases, while in contrast to the case of $I_s$, it does not improve the knowledge of the curator about $\tilde{W}$, which is already available at the input of the privacy mapping. Hence, being a bounded function, it converges. This settles the convergence of $I_s$ to $\lim_{n\to\infty} J$.

In order to obtain the limit, we proceed as follows. Consider the Markov chain $X-W-X^n-Y$, in which\footnote{َAs mentioned earlier, $X$ is a generic random variable generated in the sense that $(X,X_1,X_2,\ldots,X_n)$ are i.i.d. conditioned on $W$ according to $p_{X|W}$.} $Y\independent X_i,\ \forall i\in[n]$. The following Lemma states that requiring perfect sample privacy over a large available dataset guarantees \textit{almost} perfect sample privacy for those data samples that are not available at the input of the privacy mapping.

\begin{lemma}\label{indconverge}
In the Markov chain $X-W-X^n-Y$, where $Y\independent X_i,\forall i\in[n]$, we have
\begin{equation}
    \lim_{n\to\infty}I(X;Y)=0.
\end{equation}
\end{lemma}
\begin{proof}
The proof is provided in Appendix \ref{app2}.
\end{proof}

The last stage needed in the proof is provided in the following Lemma.
\begin{lemma}\label{lem1}
$C_1(\alpha)$ is continuous in $\alpha$.
\end{lemma}
\begin{proof}
The proof is provided in Appendix \ref{app123}.
\end{proof}

From Lemma \ref{indconverge}, we can write $I(X;Y)=\theta_n$, in which, $\lim_{n\to\infty}\theta_n=0$. Therefore, from the definition of $C_1(\cdot)$, we have
\begin{equation}
    I_s(W,X^n)\leq C_1(\theta_n).
\end{equation}
It is also evident that $J(W,X^n)\geq C_1(0)$, since the maximizer in $C_1(0)$ can be regarded as a suboptimal mapping in $J$. Therefore, we can write
\begin{equation}
    C_1(0)-\big(J(W,X^n)-I_s(W,X^n)\big)\leq I_s(W,X^n)\leq C_1(\theta_n).
\end{equation}
Finally, from Lemma \ref{lem1} and \ref{lem2}, we have
\begin{equation}
    \lim_{n\to\infty}I_s(W,X^n)=C_1(0).
\end{equation}
This completes the proof.
\end{proof}
\begin{remark}\label{rem5}
The quantity $J(W,X^n)$, defined in (\ref{jdef}), serves as a bridge between the full data observation and output perturbation models. In other words,
\begin{equation}
    C_2(0)=J(W,X^1)\geq J(W,X^2)\geq \ldots\geq J(W,X^\infty)=C_1(0),
\end{equation}
where $J(W,X^\infty)\triangleq\lim_{n\to\infty}J(W,X^n).$
\end{remark}

\begin{proposition}\label{prop5}
We have
\begin{equation}
   (C_X(W)-\log|\mathcal{X}|)^+\leq C_1(0)\leq C_2(0) \leq H(\tilde{W}|X)\leq H(W|X), 
\end{equation}
where $(x)^+\triangleq\max\{0,x\}$.
\begin{proof}
The fact that $C_1(0)\leq C_2(0)$ is immediate. For the last two inequalities, we proceed as follows. In $X-\tilde{W}-Y$ with $X\independent Y$, we have
\begin{align}
  I(\tilde{W};Y)&=I(\tilde{W},X;Y)-I(X;Y|\tilde{W}) \nonumber\\
   &=I(\tilde{W};Y|X)-I(X;Y|\tilde{W})\nonumber\\
  &= H(\tilde{W}|X)-H(\tilde{W}|Y,X)-I(X;Y|\tilde{W})\label{strictornot0}\\
  &\leq H(\tilde{W}|X)\nonumber\\
  &\leq H(W|X),\label{strictornot4}
\end{align}
where (\ref{strictornot4}) follows from having $I(W;X)=I(\tilde{W};X)$, and $H(\tilde{W})\leq H(W)$.

For the first inequality, we note that as in the proof of Lemma \ref{lemma2}, having $Y\independent X$ in the Markov chain $X-\tilde{W}-Y$ is equivalent to having $\mathbf{p}_{\tilde{W}|y}\in\textnormal{Null}(\mathbf{P}_{X|\tilde{W}}),\ \forall y\in\mathcal{Y}$. As a result, the evaluation of $C_1(0)$ reduces to the minimization of $H(\tilde{W}|Y)$ over $\{\mathbf{p}_{\tilde{W}|y}\in\textnormal{Null}(\mathbf{P}_{X|\tilde{W}})\}$ such that the marginal pmf of $\tilde{W}$ is preserved. Similarly to the proof of Corollary \ref{cor1}, we have the upper bound of $\log \left(\textnormal{rank}(\mathbf{P}_{X|\tilde{W}})\right)$ on the minimum value of $H(\tilde{W}|Y)$, such that $X\independent Y$, which in turn is upper bounded by $\log|\mathcal{X}|$. By noting that $H(\tilde{W})=C_X(W)$, the proof of the first inequality is complete.
\end{proof}
\end{proposition}
\begin{corollary}
If $W$ and $X$ can be written as $W=(W',V)$ and $X=(X',V)$, in which $W'$ and $X'$ are conditionally independent given $V$, then we have $I_s(W,X^n)=0, \forall n\geq 1$.
\begin{proof}
It can be readily verified that in this case, we have $\tilde{W}=V$, and hence, $H(\tilde{W}|X)=0.$ From Proposition \ref{prop5}, and Remark \ref{rem5}, we have $I_s(W,X^n)\leq J(W,X^n)\leq C_2(0)=0,\forall n\geq 1$.
\end{proof}
\end{corollary}

\section{Synergistic self-disclosure}\label{sec:synergistic_selfdisc}

In some cases there is no clear latent variable of interest, and the goal of the dataset owner is just to disclose as much of the dataset as possible while keeping the privacy constrains. This section studies this case for large datasets.

\subsection{Definitions and fundamental properties}

\begin{definition}
The synergistic self-disclosure capacity is defined as
\begin{equation}\label{def2}
\hat{I}_\text{s}(X^n) \triangleq \sup_{p_{Y|X^n}\in \mathcal{A}} I(X^n;Y).
\end{equation}
Similarly, the synergistic self-disclosure efficiency is defined by $\hat{\eta}(X^n) = \hat{I}_\text{s}(X^n) / H(X^n)$.
\end{definition}

One interesting property of $\hat{I}_\text{s}(X^n)$ is that $\hat{I}_\text{s}(X^n) \geq I_\text{s}(W,X^n)$ for any latent feature $W$, being this a direct consequence of the data processing inequality applied to $W-X^n-Y$. We now provide a simple upper bound for the synergistic self-disclosure efficiency.

\begin{lemma}\label{lemma:self-disclosure_bound}
The following upper bound holds:
\begin{equation*}
  \hat{\eta}(X^n) \leq 1 - \frac{ \max_{j\in [n]} H(X_j)  }{  H(X^n) }
\end{equation*}
\end{lemma}
\begin{proof}
This follows directly from Proposition~\ref{pr:upper_bound}, by setting $W=X^n$. 
\end{proof}

The previous lemma shows that $\hat{\eta}(X^n) < 1$ for any finite dataset, i.e., finite $n$. Hence, one might wonder if there are cases in which $\hat{\eta}\to 1$ as $n$ grows. Our next result shows that, remarkably, this happens whenever the entropy rate of the process, denoted by $H(\mathcal{X})$, exists and is non-zero. 
\begin{theorem}\label{th2}
Consider a stochastic process $\{X_i\}_{i\geq 1}$, with $|\mathcal{X}_i|\leq M<\infty,\ \forall i\geq 1$. If the entropy rate of this process exists, then
\begin{equation}\label{asres1}
\lim_{n\to\infty} \frac{\hat{I}_\text{s}(X^n)}{n}  = H(\mathcal{X}),
\end{equation}
where $H(\mathcal{X})$ denotes the entropy rate of the stochastic process $\{X_i\}_{i\geq 1}$. Furthermore, if $H(\mathcal{X})\neq 0$, we have
\begin{equation}
\lim_{n\to\infty} \hat{\eta}(X^n) = 1.
\end{equation}
\end{theorem}
\begin{proof}
From Corollary \ref{cor1}, we have
\begin{equation}
\min_{p_{Y|X^n} \in \mathcal{A}}\!\!\! H(X^n|Y) \leq \log(nM),
\end{equation}
which results in
\begin{equation}
\frac{H(X^n)}{n}-\frac{\log(nM)}{n} \leq \frac{\hat{I}_\text{s}(X^n)}{n} 
\leq 
\frac{H(X^n)}{n}\label{lowb}.
\end{equation}
Taking the limit $n\to\infty$ proves (\ref{asres1}). Hence, when $H(\mathcal{X})\neq 0$, we can write
\begin{equation}
\lim_{n\to\infty} \hat{\eta}(X^n) = 
\lim_{n\to\infty} \frac{\frac{\hat{I}_\text{s}(X^n)}{n}}{\frac{H(X^n)}{n}} =\frac{H(\mathcal{X})}{H(\mathcal{X})} = 1.
\end{equation}
%
%
\end{proof}
Theorem \ref{th2} signifies the fact that the constraints of perfect sample privacy, i.e., $Y\independent X_i,\ \forall i\in[n]$, result in no asymptotic loss of the disclosure.

\begin{corollary}
Assume that instead of the perfect sample privacy constraint, i.e., $Y\independent X_i,\ \forall i\in[n]$, a more restrictive constraint is used, such as $Y\independent (X_i,X_{i+1}),\ \forall i\in[n-1]$, and in general, $Y\independent (X_i,\ldots,X_{i+k-1}),\ \forall i\in[n-k+1]$ for a fixed (i.e., not scaling with $n$) positive integer $k$. The results of Theorem \ref{th2} still hold under these conditions.
\begin{proof}
Having $Y\independent (X_i,\ldots,X_{i+k-1}),\ \forall i\in[n-k+1]$, the number of rows of matrix $\mathbf{P}$ is at most $M^k(n-k+1)$, which is an upper bound on its rank. Since in the evaluation of self-disclosure capacity, the extreme points of $\mathbb{S}$ have at most $M^k(n-k+1)$ non-zero elements, and $\frac{\log(M^k(n-k+1))}{n}\to 0$ as $n\to\infty$, the proof of Theorem \ref{th2} remains unaltered.
\end{proof}
\end{corollary}

The next example illustrates how the efficiency of synergistic disclosure can converge to 1 as $n$ grows even with a more stringent privacy constraint compared to perfect sample privacy.

\begin{example}\label{ex:4}
Let $X^n$ be a dataset of i.i.d. r.v.'s that are uniformly distributed over $[M]$ for some positive integer $M$. For a fixed $k\in[n]$, let $L^{n-k+1}$ be a sequence of r.v.'s which are generated according to $L_i\triangleq Q+X_i+X_{i+1}+\ldots+X_{i+k-1}(\textnormal{mod}\ M),\ \forall i\in[n-k+1]$, where $Q$ is an auxiliary r.v. which is uniform on $[M]$, and independent of the dataset $X^n$. One can see that $L_1,\ldots,L_{n-k+1}$ are also i.i.d. and uniform on $[M]$. Let $Y\triangleq L^{n-k+1}$. One can verify that the conditional pmf of $Y$ conditioned on $(X_i,\dots,X_{i+k-1})$ is not affected by any realization of the tuple for all $i\in[n-k+1]$, which guarantees the strengthened privacy constraint that $Y\independent (X_i,\dots,X_{i+k-1}),\ \forall i\in[n-k+1]$. Finally, one can check that
\begin{align}
\lim_{n\rightarrow\infty} \eta(X^n) &\geq \lim_{n\rightarrow\infty} \frac{I(X^n;Y)}{H(X^n)}\nonumber\\
&=\lim_{n\rightarrow\infty} \frac{I(X^n;Y|Q)-I(X^n;Q|Y)}{H(X^n)}\label{indQ}\\
&\geq \lim_{n\rightarrow\infty} \frac{H(Y|Q)-\log M}{H(X^n)}\label{recovY1}\\
&\geq \lim_{n\rightarrow\infty} \frac{H(Y|Q,X_1^{k-1})-\log M}{H(X^n)}\label{recovY2}\\
&= \lim_{n\rightarrow\infty} \frac{H(X_k^n)-\log M}{H(X^n)}\label{recovY3}\\
&= \lim_{n\rightarrow\infty} \frac{(n-k+1)M-\log M}{nM}\nonumber\\
&= 1,\label{recovY4}
\end{align}
where (\ref{indQ}) follows from having $Q\independent X^n$; (\ref{recovY1}) follows from the fact that $Y$ is a deterministic function of the tuple $(Q,X^n)$, and $I(X^n;Q|Y)\leq H(Q)=\log M$; (\ref{recovY1}) is from the fact that conditioning does not increase entropy, and also we set $X_1^0\triangleq \emptyset$; (\ref{recovY3}) is due to the fact that conditioned on $(Q,X_1^{k-1})$, $Y$ and $X_k^n$ have a one-to-one correspondence. Finally, in (\ref{recovY4}), we assume that $k$ satisfies the constraint $\lim_{n\to\infty} \frac{k}{n}=0$.   
\end{example}

\subsection{Self-disclosure of continuous variables}\label{sec:cont}

Here we study the self-disclosure properties of small datasets composed of two continuous variables $X_1,X_2$. 

\begin{theorem}\label{th6}
Let $X_1,X_2$ be two independent and continuous random variables with $\mathcal{X}_1,\mathcal{X}_2\subset\mathbb{R}$. We have
\begin{equation}
\sup_{\substack{Y|X_1,X_2:\\Y\independent X_1,\ Y\independent X_2}}I(X_1,X_2;Y)=\infty
\end{equation}
\end{theorem}
\begin{proof}
Let $K$ be an arbitrary positive integer. Partition $\mathcal{X}_1$ into $K$ disjoint intervals $\mathcal{I}_i,\ (i\in[K])$ with equal probabilities, i.e., $\mbox{Pr}\{X_1\in\mathcal{I}_i\}=\frac{1}{K},\forall i\in[K]$. Similarly, partition $\mathcal{X}_2$ into $K$ disjoint intervals $\mathcal{J}_i,\ (i\in[K])$ with equal probabilities. Let $\hat{Y}$ be a deterministic function of $(X_1,X_2)$ defined as
\begin{equation}
    \hat{Y} = (i+j)(\mbox{mod }K)+1 ,\ \mbox{if } (x_1,x_2)\in\mathcal{I}_i\times\mathcal{J}_j,\ \forall i,j\in[K].
\end{equation}
It is easy to verify that $\hat{Y}$ is uniformly distributed over $[K]$. Also, the distribution of $\hat{Y}$ is the same after observing any realization of $X_1$ (or $X_2$); hence, we have $\hat{Y}\independent X_1,\ \hat{Y}\independent X_2$. By definition,
\begin{align}
 \sup_{\substack{Y|X_1,X_2:\\Y\independent X_1,\ Y\independent X_2}}I(X_1,X_2;Y)&\geq I(X_1,X_2;\hat{Y})=\log K,
\end{align}
where we have used the fact that $H(\hat{Y}|X_1,X_2)=0$, since $\hat{Y}$ is a deterministic function of $(X_1,X_2)$. Letting $K\to\infty$ completes the proof.
\end{proof}

\subsection{Heuristic approaches}\label{sec:heu}

In the method proposed in Section~\ref{sec:3}, the complexity of the computations required for building the optimal synergistic mapping grows exponentially with the size of the dataset. The main bottleneck of Algorithm~\ref{alg2} is the exhaustive search over groups of columns of $\mathbf{A}$ that is needed to find all the extreme points of $\mathbb{S}$. A straightforward solution to this issue is to only explore a fixed number of groups of columns, to be chosen randomly. Although this approach generates a mapping that satisfies perfect sample privacy, numerical evaluations show that it performance tends to zero if the chosen number of evaluated groups of columns is bounded. Therefore, more ingenious heuristic methods for building suboptimal mappings are needed.

It is worth to note that the search of extreme points of $\mathbb{S}$ can become computationally expensive due to two reasons: either the dataset has a large number of samples, or their alphabet is big. In the sequel, Section~\ref{sec:heu1} presents a procedure that addresses the first issue, while Section~\ref{sec:heu2} takes care of the second. Please note that, although both procedures are presented for scenarios where the datases are composed of independent samples, it is straightforward to generalize them to datasets composed of groups of samples that are independent of other groups, and can be jointly processed.

\subsubsection{Partial processing method}
\label{sec:heu1}

Let us assume that $X^n$ is composed by independent samples. Let us generate mappings of the form $p_{Y_{\text{par},j}|X_j,X_{j+1}}$ for $j\in[n-1]$ according to 
\begin{equation}
    p_{Y_{\text{par},j}|X_j,X_{j+1}} = \argmax_{\substack{p_{Y|X_j,X_{j+1}}: \\Y\independent X_j, Y\independent X_{j+1}}} I(X_j,X_{j+1};Y), \label{eq:heu_parwise}
\end{equation}
Note that the complexity of building $Y_\text{par}^{n-1}=(Y_{\text{par},1},\dots,Y_{\text{par},n-1})$ scales linearly with the size of the dataset. The next lemma shows that $Y_\text{par}^{n-1}$ guarantees perfect sample privacy for $X^n$.
\begin{lemma}\label{lemma:par}
With the above construction, $X_k\independent Y_\text{par}^{n-1}$ for all $k\in [n]$. Moreover, the performance of this strategy is 
\begin{equation}
I(Y_\text{par}^{n-1};X^n)  = \sum_{j=1}^{n-1}  I(Y_{\text{par},j};X_j | X_{j+1}).
\end{equation}
\begin{proof}
See Appendix~\ref{app:heuristic_par}.
\end{proof}
\end{lemma}

\begin{example}
Consider the case where $X^n$ are i.i.d. samples. Then, the asymptotic disclosure efficiency for this mapping, denoted by $\eta_\text{par}$, is found to be monotonically increasing with limit given by
\begin{equation}
\lim_{n\to\infty} \hat{\eta}_\text{par}(X^n) = \lim_{n\to\infty} \frac{(n-1) I(Y_{\text{par},1};X_1|X_2)}{n H(X_1)} = \frac{I(Y_{\text{par},1};X_1|X_2)}{H(X_1)}.\nonumber 
\end{equation}
The performance of this disclosure mapping for the case of i.i.d. Bernoulli samples is illustrated in Figures~\ref{fig:heu_disclosure} and Figure~\ref{fig:heu_efficiency}. The asymptotic efficiency is maximal only for the case of $\mathbb{P}\{X_1=1\}=1/2$.
\end{example}

\begin{figure}[h]
  \centering
\begin{tikzpicture}[scale=0.6]
  \begin{axis}
[
 mark size=0.7pt,
 legend pos=north west,
 cycle list name = nano line style,
 xlabel={Dataset distribution $\mathbb{P}\{ X_k=1\}$},
 xmax = 0.5,
 xmin =0,
 ylabel={Disclosure ($I_\text{s}$)},
 ymin =0,
 ymax = 3,
 grid=both,
 minor grid style={gray!25},
 major grid style={gray!25},
 width=0.7\linewidth,
 height=0.55\linewidth,
 tick label style={/pgf/number format/fixed},
 no marks
] 
\addplot table[x=q, y = Optimal, col sep=comma]{./heu_data_1.csv} ; 
\addlegendentry{Optimal mapping} ; 
\addplot table[x=q, y = Loc, col sep=comma]{./heu_data_1.csv} ; 
\addlegendentry{Partial-processing method} ; 
\addplot table[x=q, y = Pre, col sep=comma]{./heu_data_1.csv} ;
\addlegendentry{Pre-processing method} ;
\end{axis}
\end{tikzpicture}
\caption{Performance of two heuristic approaches versus the optimal scheme, for the case of $n=4$ i.i.d. Bernoulli data samples with parameter $\text{Pr}\{ X_k=1\}$ represented on the x-axis.}
\label{fig:heu_disclosure}
\end{figure}
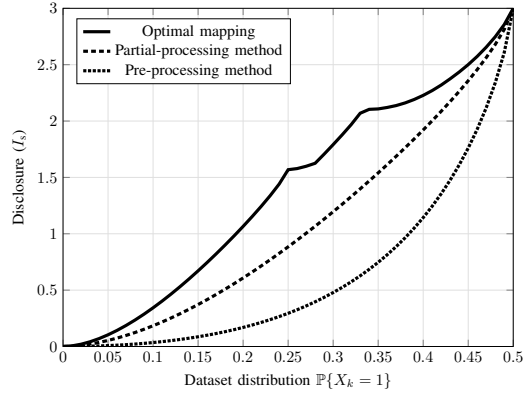

\begin{figure}[h]
  \centering
\begin{tikzpicture}[scale=0.6]
  \begin{axis}
[
 mark size=0.7pt,
 legend pos=south east,
 cycle list name = nano line style,
 xlabel={Dataset distribution ($\mathbb{P}\{ X_k=1\}$)},
 xmax = 0.5,
 xmin =0,
 ylabel={Asymptotic disclosure efficiency},
 ymin =0,
 ymax = 1.05,
 grid=both,
 minor grid style={gray!25},
 major grid style={gray!25},
 width=0.7\linewidth,
 height=0.55\linewidth,
 tick label style={/pgf/number format/fixed},
 no marks
] 
\addplot table[x=Pw, y = Optimal, col sep=comma]{./heu_data_3.csv} ; 
\addlegendentry{Optimal mapping} ; 
\addplot table[x=Pw, y = Loc, col sep=comma]{./heu_data_3.csv} ; 
\addlegendentry{Local-processing method} ; 
\addplot table[x=Pw, y = Pre, col sep=comma]{./heu_data_3.csv} ;
\addlegendentry{Pre-processing method} ;
\end{axis}
\end{tikzpicture}
\caption{Disclosure efficiency for heuristic approaches v/s optimal efficiency -- given by Theorem~\ref{th2} -- for asymptotically large datasets of binary i.i.d. samples.}
\label{fig:heu_efficiency}
\end{figure}
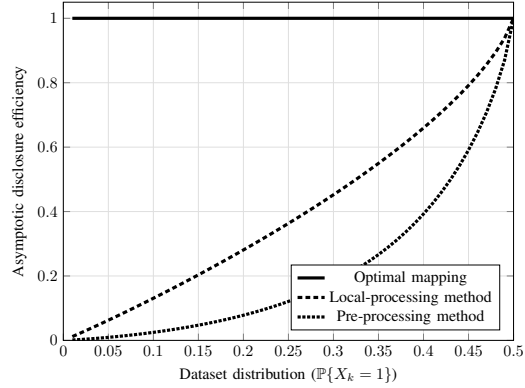

This approach can be further generalized as follows: for a given $k\in[n-1]$, build mappings $p_{Y_{\text{par}(k),j}|X_j,\dots,X_{j+k-1}}$ for $j\in[n-k+1]$ given by
\begin{equation}
    p_{Y_{\text{par}(k),j}|X_j,\dots,X_{j+k-1}} = 
    \argmax_{\substack{p_{Y|X_j,\dots,X_{j+k-1}}: \\
    Y\independent X_i, i=j,\dots,j+k-1}} I(X_j,\dots,X_{j+k-1};Y). \nonumber
\end{equation}
Note that $\eqref{eq:heu_parwise}$ correspond to the case of $k=2$. Following a proof entirely analogous to the one of Lemma~\ref{lemma:par}, one can show that $Y^{n-k+1}_{\text{par}(k)} = (Y_{\text{par}(k),1},\dots,Y_{\text{par}(k),n-k+1})$ satisfies perfect sample privacy for all $k\in[n-1]$. Interestingly, schemes with large values of $k$ attain high disclosure efficiency, at the cost of incurring in more expensive computations for calculating the corresponding mappings. Therefore, $k$ can be tuned in order to balance efficiency and computational complexity.

\subsubsection{Pre-processing of independent datasets}
\label{sec:heu2}
Another approach for building sub-optimal disclosure mappings is to perform a pre-processing stage over the dataset, in order to make it reach a distribution for which the optimal disclosure strategy is already known. Furthermore, this pre-processing must be carried out in a way that does not violate the privacy constraints.

Assume that $X^n$ is a dataset of independent variables with the same alphabet, i.e., $X_i\in\mathcal{X}(=[|\mathcal{X}|]),\forall i\in[n]$. If all $X_i$'s are uniformly distributed, then the optimal solution is $Y=L^{n-1}$, in which $L_i=X_i+X_{i+1} (\textnormal{ mod } |\mathcal{X}|)$; however, if the marginal distributions, i.e., those of $X_i$'s, are not uniform, then the aforementioned disclosure mapping does not satisfy perfect sample privacy anymore, and the optimal solution is obtained via the procedure explained in Section \ref{sec:3}. One sub-optimal solution here is to first pass each $X_i$ through a uniformizer, i.e., a pre-processing mapping, denoted by $p_{S_i|X_i}$, such that $S_i$ is uniform over $\mathcal{X}$. Then, the optimal synergistic disclosure mapping for i.i.d. uniformly distributed data samples can be applied to the new dataset $S^n$. In this context, we denote the output by $Y_\text{pre}$. Finally, the fact that $X_i-S_i-Y_{\text{pre}}$ form a Markov chain in conjunction with $Y_\text{pre}\independent S_i$ proves $Y_{\text{pre}}\independent X_i$. This procedure is illustrated in the next example.

\begin{example}
Assume that $X_i$'s are i.i.d $\text{Bern}(q)$ with $q\in (0,1/2]$. Following the previous discussion, the uniformizer, i.e., $p_{S_j|X_j}$, is a $Z$ channel with crossover probability $\beta = (0.5-q)/(1-q)$, as shown in Figure~\ref{fig:pre_heuristic}. With this construction, we have a new dataset composed of i.i.d. $\text{Bern}(\frac{1}{2})$ samples $S_i$. Finally, we set $Y_{\textnormal{pre}}=[S_1\oplus S_2,\ S_2\oplus S_3, \ldots ,S_{n-1}\oplus S_n]^T$. The performance of this strategy is obtained as follows.
\begin{align}
  I(Y_\text{pre};X^n)&= \sum_{k=1}^{n-1} I(S_k\oplus S_{k+1};X_k,X_{k+1})\nonumber\\
  &=\sum_{k=1}^{n-1} H(S_k\oplus S_{k+1})-H(S_k\oplus S_{k+1}|X_k,X_{k+1})\nonumber\\
  &=\sum_{k=1}^{n-1} \bigg(1-\sum_{(x_k,x_{k+1})\in\{0,1\}^2}H(S_k\oplus S_{k+1}|X_k=x_k,X_{k+1}=x_{k+1})\bigg)\nonumber\\
  &=(n-1) \Big[ 1 - 2 q (1-q) h_b( \beta) - (1-q)^2 h_b\big( 2\beta(1-\beta)\big) \Big]\label{la},
\end{align}
where $h_b(p) \triangleq - p \log p - (1-p) \log (1-p)$ denotes the binary entropy function, and the last step follows from the fact that $H(S_k\oplus S_{k+1}|X_k=0,X_{k+1}=1)=H(S_k\oplus S_{k+1}|X_k=1,X_{k+1}=0)=h_b(\beta), H(S_k\oplus S_{k+1}|X_k=0,X_{k+1}=0)=h\big( 2\beta(1-\beta)\big)$, and $H(S_k\oplus S_{k+1}|X_k=1,X_{k+1}=1)=0$. Finally, the disclosure efficiency of this method, $\hat{\eta}_\text{pre}(X^n)$, grows monotonically with $n$, with the asymptotic disclosure efficiency being as
\begin{equation}
\lim_{n\to\infty} \hat{\eta}_\text{pre}(X^n) =  \frac{1}{{h(q)}} \Big[ 1 - 2 q (1-q) h( \beta) - (1-q)^2 h\big( 2\beta(1-\beta)\Big] . \nonumber
\end{equation}
The performance of this algorithm is shown in Figure~\ref{fig:heu_disclosure} and \ref{fig:heu_efficiency} for the case of $n=4$.
\end{example}
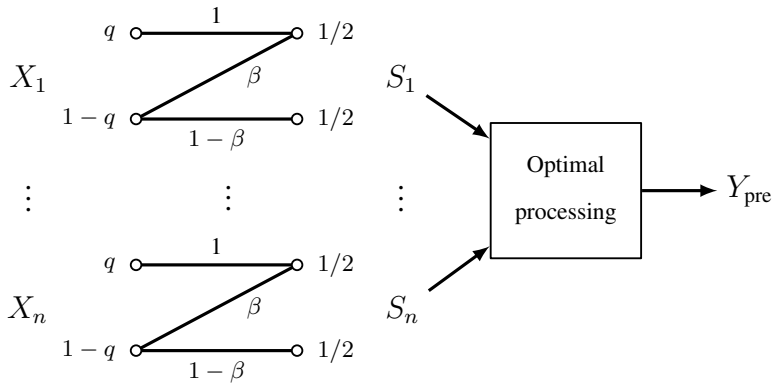
\begin{figure}\label{fig:pre_heuristic}
\begin{center}
\begin{tikzpicture}
  \node[dspnodeopen, minimum width=4pt, dsp/label=left] (X1_1) {\small{$q$}};    
  \node[dspnodeopen, minimum width=4pt, below=1cm of X1_1, dsp/label=left] (X1_0) {\small{$1-q$}};
  \node[dspnodeopen, minimum width=4pt, right=2cm of X1_1, dsp/label=right] (S1_1) {\small{$1/2$}};    
  \node[dspnodeopen, minimum width=4pt, below=1cm of S1_1, dsp/label=right] (S1_0) {\small{$1/2$}};
  \node[below left=0.25cm and 1cm of X1_1](X1){\large{$X_1$}};
  \node[below right=0.25cm and 1cm of S1_1](S1){\large{$S_1$}};
  \node[below=0.8cm of X1](dots1){\large{$\vdots$}};
  \node[below right=0.5cm and 1cm of X1_0](dots2){\large{$\vdots$}};
  \node[below=0.8cm of S1](dots3){\large{$\vdots$}};
  \node[dspnodeopen, minimum width=4pt, below=1.8cm of X1_0, dsp/label=left] (Xn_1) {\small{$q$}};    
  \node[dspnodeopen, minimum width=4pt, below=1cm of Xn_1, dsp/label=left] (Xn_0) {\small{$1-q$}};
  \node[dspnodeopen, minimum width=4pt, right=2cm of Xn_1, dsp/label=right] (Sn_1) {\small{$1/2$}};    
  \node[dspnodeopen, minimum width=4pt, below=1cm of Sn_1, dsp/label=right] (Sn_0) {\small{$1/2$}};
  \node[below left=0.25cm and 1cm of Xn_1](Xn){\large{$X_n$}};
  \node[below right=0.25cm and 1cm of Sn_1](Sn){\large{$S_n$}};
  \node[dspsquare, right=1cm of dots3, minimum size=1.8cm, text height=-0.4cm, text width=2cm, align=center](box){\small{Optimal processing}};
  \node[right=1cm of box](Y){\large{$Y_\text{pre}$}};
\draw[,line width=1.2pt] (X1_1) to node[midway, above]{\small{1}} (S1_1);
\draw[,line width=1.2pt] (X1_0) to node[near end, below]{\small{$\beta$}} (S1_1);
\draw[,line width=1.2pt] (X1_0) to node[midway, below]{\small{$1-\beta$}} (S1_0);
\draw[,line width=1.2pt] (Xn_1) to node[midway, above]{\small{1}} (Sn_1);
\draw[,line width=1.2pt] (Xn_0) to node[near end, below]{\small{$\beta$}} (Sn_1);
\draw[,line width=1.2pt] (Xn_0) to node[midway, below]{\small{$1-\beta$}} (Sn_0);
\draw[-latex,line width=1.2pt] (S1) to node{} (box);
\draw[-latex,line width=1.2pt] (Sn) to node{} (box);
\draw[-latex,line width=1.2pt] (box) to node{} (Y);
\end{tikzpicture}
\end{center}
\caption{Two-step process to generate a suboptimal disclosure mapping. Each sample is first pre-processed via a Z-channel with $\beta=(1/2-q)/(1-q)$; then, an optimal processing is performed over $S^n$.}
\end{figure}

When processing samples with small alphabets this approach is often less efficient than the one described in Section~\ref{sec:heu1}. However, the main strength of this approach is that it can be applied to datasets composed by samples with large alphabets, e..g. using the mapping outlined in Example~\ref{ex:4}.

\section{Conclusions}\label{sec:conclusions}

This work develops methods to enable synergistic data disclosure, which allow to make publicly available collective properties of a dataset while keeping the values of each data sample perfectly confidential. The coexistence of privacy and utility is attained by exploiting counter-intuitive properties of multivariate statistics, which allow a variable to be correlated with a random vector while being independent to each of it components. An algorithm has been presented to build an optimal synergistic disclosure mapping following standard LP techniques. Moreover, we developed closed-form expressions for the synergistic disclosure capacity in a number of scenarios. 

While perfect sample privacy could seem to be a restrictive ideal, our results show that in many scenarios there exist disclosure mappings whose efficiency tends asymptotically to one. This means that the amount of data that one needs to hide in order to guarantee perfect sample privacy becomes negligible for large datasets. This promising result -- which holds with remarkable generality -- shows that perfect sample privacy can be extremely efficient while providing strong privacy guarantees. 

When compared with differential privacy, both approaches share the property of being robust to post-processing (as any function of a perfect sample-private mapping keeps this property). 
An advantage of our approach is that, while differential privacy is known to be less efficient in cases of correlated data~\cite{liu2016} (although partial solutions to this issue have been proposed~\cite{liu2016,gehrke2011towards,li2013membership,kifer2014pufferfish}), our approach is well-suited to data with any underlying distribution. However, a limitation of our approach is that it requires knowledge of the statistics of the dataset and latent feature, which are unknown in many real scenarios. The estimation of unknown statistics can be approached by using well-established methods of Bayesian
inference~\cite{gelman2013bayesian} and machine learning~\cite{Bishop:2006:PRM:1162264}. It is, however, an important future step is to study how estimation errors could impact the privacy guarantees.



\appendices
\section{}\label{app1}
Let $\mathcal{Y}$ be an arbitrary set. Let $\mathbb{S}$ be the set of probability vectors defined in (\ref{poly}). Let $\mathcal{Q}$ denote an index set of $\mbox{rank}(\mathbf{P})$ linearly independent columns of $\mathbf{P}$. Hence, the columns corresponding to the index set $\mathcal{Q}^c=[|\hat{\mathcal{X}}|]\backslash \mathcal{Q}$ can be written as a linear combination of the columns indexed by $\mathcal{Q}$. Let $\pi:[\mbox{nul}(\mathbf{P})]\to\mathcal{Q}^c$ such that $\pi(i)<\pi(j)$ for $i<j,\forall i,j\in [\mbox{nul}(\mathbf{P})]$.  Let $\mathbf r :\mathbb{S}\to\mathbb{R}^{\mbox{nul}(\mathbf{P})+1}$ be a vector-valued mapping defined element-wise as
\begin{align}
r_{i}(\mathbf{p})&=\mathbf{p}(\pi(i)),\ \forall i\in[\mbox{nul}(\mathbf{P})]\nonumber\\
r_{\mbox{nul}(P)+1}(\mathbf{p})&=H(\mathbf{P}_{W|X}\mathbf{p}),\label{cara}
\end{align}
where $\mathbf{p}(\pi(i))$ denotes the $\pi(i)$-th element of the probability vector $\mathbf{p}$. Since $\mathbb{S}$ is a closed and bounded subset of $\mathcal{P}(\hat{\mathcal{X}})$, it is compact. Also, $\mathbf{r}$ is a continuous mapping from $\mathbb{S}$ to $\mathbb{R}^{\mbox{nul}(P)+1}$. Therefore, from the support lemma \cite{Elgamal}, for every $Y\sim F(y)$ defined on $\mathcal{Y}$, there exists a random variable $Y'\sim p(y')$ with $|\mathcal{Y'}|\leq \mbox{nul}(\mathbf{P})+1$ and a collection of conditional probability vectors $\mathbf{p}_{X|y'}\in\mathbb{S}$ indexed by $y'\in\mathcal{Y}'$, such that
\begin{equation*}
\int_{\mathcal{Y}}r_i(\mathbf{p}_{X|y})dF(y)=\sum_{y'\in\mathcal{Y'}}r_i(\mathbf{p}_{X|y'})p(y'),\ i\in[\mbox{nul}(\mathbf{P})+1].
\end{equation*}
It can be verified that by knowing the marginals $\mathbf{p}_{X_i}, \forall i\in[n]$, and the $\mbox{nul}(\mathbf{P})$ elements of $\mathbf{p}_{X^n}$ corresponding to index set $\mathcal{Q}^c$, the remaining $|\hat{\mathcal{X}}|-\mbox{nul}(\mathbf{P})$ elements of $\mathbf{p}_{X^n}$ can be uniquely identified. Therefore, for an arbitrary $Y$ in $W-X^n-Y$, that satisfies $X_i\independent Y, \forall i\in[n]$, the terms $p_X(\cdot)$, and $I(W;Y)$ are preserved if $Y$ is replaced with $Y'$. So are the conditions of independence as $\mathbf{p}_{X|Y'}\in\mathbb{S}, \forall y'\in\mathcal{Y}'$. Since we can simply construct the Markov chain $W-X^n-Y'$, there is no loss of optimality in considering $|\mathcal{Y}|\leq\mbox{nul}(\mathbf{P})+1$.

The attainability of the supremum follows from the continuity of $I(W;Y)$ and the compactness of $\mathbb{S}$.


\section{}\label{app1}
Let $p^*_{Y|\tilde{W},X^n}$ denote the maximizer of (\ref{jdef}), which induces $p^*_{Y|X^n}$ as
\begin{equation*}
    p^*_{Y|X^n}(y|x^n)=\sum_{\tilde{w}}p(\tilde{w}|x^n)p^*_{Y|\tilde{W},X^n}(y|\tilde{w},x^n),\ \forall y,x^n.
\end{equation*}
It is evident that when $Y$ is generated by applying $p^*_{Y|X^n}$ to $X^n$, it satisfies the perfect sample privacy constraints, i.e., $Y\independent X_i, \forall i\in[n].$ Let $q(\tilde{w},x^n,y)\triangleq p(\tilde{w},x^n)p^*_{Y|X^n}(y|x^n)$, and $p^*(\tilde{w},x^n,y)\triangleq p(\tilde{w},x^n)p^*_{Y|\tilde{W},X^n}(y|\tilde{w},x^n)$.
We first show that
\begin{equation}\label{sad}
    \lim_{n\to\infty}d_{TV}(q,p^*)=\lim_{n\to\infty}\|q-p^*\|_1=0,
\end{equation}
where $d_{TV}$ denotes the total vatiation distance.
The proof is as follows.
\begin{align}
    d_{TV}(q,p^*)&=\sum_{\tilde{w},x^n,y}|q(\tilde{w},x^n,y)-p^*(\tilde{w},x^n,y)|\nonumber\\
    &=\sum_{\tilde{w},x^n,y}p(\tilde{w},x^n)|p^*(y|x^n)-p^*(y|\tilde{w},x^n)|\nonumber\\
    &=\sum_{\tilde{w}}p(\tilde{w})\sum_{x^n}p(x^n|\tilde{w})\sum_y|p^*(y|x^n)-p^*(y|\tilde{w},x^n)|\nonumber\\
    &=\sum_{\tilde{w}}p(\tilde{w})\sum_{x^n\in\mathcal{T}^n_{\epsilon}(p_{X|\tilde{W}}(\cdot|\tilde{w}))}p(x^n|\tilde{w})\sum_y|p^*(y|x^n)-p^*(y|\tilde{w},x^n)|\nonumber\\
    &\ \ \ + \sum_{\tilde{w}}p(\tilde{w})\sum_{x^n\not\in\mathcal{T}^n_{\epsilon}(p_{X|\tilde{W}}(\cdot|\tilde{w}))}p(x^n|\tilde{w})\sum_y|p^*(y|x^n)-p^*(y|\tilde{w},x^n)|\label{tv1}\\
    &=\sum_{\tilde{w}}p(\tilde{w})\sum_{x^n\in\mathcal{T}^n_{\epsilon}(p_{X|\tilde{W}}(\cdot|\tilde{w}))}p(x^n|\tilde{w})\sum_y|p^*(y|x^n)-p^*(y|\tilde{w},x^n)|\nonumber\\
    &\ \ \ + 2\sum_{\tilde{w}}p(\tilde{w})\sum_{x^n\not\in\mathcal{T}^n_{\epsilon}(p_{X|\tilde{W}}(\cdot|\tilde{w}))}p(x^n|\tilde{w})\label{tv2}\\
    &\leq \sum_{\tilde{w}}p(\tilde{w})\sum_{x^n\in\mathcal{T}^n_{\epsilon}(p_{X|\tilde{W}}(\cdot|\tilde{w}))}p(x^n|\tilde{w})\sum_y|p^*(y|x^n)-p^*(y|\tilde{w},x^n)|\nonumber\\
    &\ \ \ +2\textnormal{Pr}\{X^n\not\in\mathcal{T}^n_{\epsilon}(p_{X|\tilde{W}}(\cdot|\tilde{W}))\}\nonumber\\
     &=\sum_{\tilde{w}}p(\tilde{w})\sum_{x^n\in\mathcal{T}^n_{\epsilon}(p_{X|\tilde{W}}(\cdot|\tilde{w}))}p(x^n|\tilde{w})\sum_y\bigg|\sum_{j\in\tilde{\mathcal W}}p_{\tilde{W}|X^n}(j|x^n)\bigg(p^*_{Y|\tilde{W},X^n}(y|j,x^n)-p^*(y|\tilde{w},x^n)\bigg)\bigg|\nonumber\\
    &\ \ \ +2\textnormal{Pr}\{X^n\not\in\mathcal{T}^n_{\epsilon}(p_{X|\tilde{W}}(\cdot|\tilde{W}))\}\label{tv3}\\
    &\leq\sum_{\tilde{w}}p(\tilde{w})\sum_{x^n\in\mathcal{T}^n_{\epsilon}(p_{X|\tilde{W}}(\cdot|\tilde{w}))}p(x^n|\tilde{w})\sum_{j\neq\tilde{w}}p_{\tilde{W}|X^n}(j|x^n)\sum_y|p^*_{Y|\tilde{W},X^n}(y|j,x^n)-p^*(y|\tilde{w},x^n)|\nonumber\\
    &\ \ \ +2\textnormal{Pr}\{X^n\not\in\mathcal{T}^n_{\epsilon}(p_{X|\tilde{W}}(\cdot|\tilde{W}))\}\label{tv4}\\
    &\leq 2\sum_{\tilde{w}}p(\tilde{w})\sum_{x^n\in\mathcal{T}^n_{\epsilon}(p_{X|\tilde{W}}(\cdot|\tilde{w}))}p(x^n|\tilde{w})\sum_{j\neq\tilde{w}}p_{\tilde{W}|X^n}(j|x^n)+2\textnormal{Pr}\{X^n\not\in\mathcal{T}^n_{\epsilon}(p_{X|\tilde{W}}(\cdot|\tilde{W}))\}\label{tv5}\\
    &=2\sum_{\tilde{w}}p(\tilde{w})\sum_{x^n\in\mathcal{T}^n_{\epsilon}(p_{X|\tilde{W}}(\cdot|\tilde{w}))}\frac{p(x^n|\tilde{w})}{p(x^n)}\sum_{j\neq\tilde{w}}p_{\tilde{W},X^n}(j,x^n)+2\textnormal{Pr}\{X^n\not\in\mathcal{T}^n_{\epsilon}(p_{X|\tilde{W}}(\cdot|\tilde{W}))\}\nonumber\\
    &\leq2\sum_{\tilde{w}}p(\tilde{w})\sum_{x^n\in\mathcal{T}^n_{\epsilon}(p_{X|\tilde{W}}(\cdot|\tilde{w}))}\frac{1}{p(\tilde{w})}\sum_{j\neq\tilde{w}}p_{\tilde{W},X^n}(j,x^n)+2\textnormal{Pr}\{X^n\not\in\mathcal{T}^n_{\epsilon}(p_{X|\tilde{W}}(\cdot|\tilde{W}))\}\label{tv6}\\
    &=2\sum_{\tilde{w}}\sum_{j\neq\tilde{w}}\textnormal{Pr}\{X^n\in\mathcal{T}^n_{\epsilon}(p_{X|\tilde{W}}(\cdot|\tilde{w})),\tilde{W}=j\}+2\textnormal{Pr}\{X^n\not\in\mathcal{T}^n_{\epsilon}(p_{X|\tilde{W}}(\cdot|\tilde{W}))\}\\
    &\leq 2\sum_{\tilde{w}}\sum_{j\neq\tilde{w}}\textnormal{Pr}\{X^n\not\in\mathcal{T}^n_{\epsilon}(p_{X|\tilde{W}}(\cdot|j)),\tilde{W}=j\}+2\textnormal{Pr}\{X^n\not\in\mathcal{T}^n_{\epsilon}(p_{X|\tilde{W}}(\cdot|\tilde{W}))\}\label{tv7}\\
    &\leq 2\sum_{\tilde{w}}\sum_{j\in\mathcal{\tilde{W}}}\textnormal{Pr}\{X^n\not\in\mathcal{T}^n_{\epsilon}(p_{X|\tilde{W}}(\cdot|j)),\tilde{W}=j\}+2\textnormal{Pr}\{X^n\not\in\mathcal{T}^n_{\epsilon}(p_{X|\tilde{W}}(\cdot|\tilde{W}))\}\nonumber\\
    &= 2(|\mathcal{\tilde{W}}|+1)\textnormal{Pr}\{X^n\not\in\mathcal{T}^n_{\epsilon}(p_{X|\tilde{W}}(\cdot|\tilde{W}))\},\label{tv8}
\end{align}
where in (\ref{tv2}), we use the fact that for two pmfs $a(\cdot),b(\cdot)$,
\begin{equation}\label{tvbound}
  \sum_y|a(y)-b(y)|\leq 2.  
\end{equation}
(\ref{tv4}) results from the triangle inequality, and the condition $j\neq\tilde{w}$ comes from the fact that in (\ref{tv3}), the term $p^*_{Y|\tilde{W},X^n}(y|j,x^n)-p^*(y|\tilde{w},x^n)$ is zero when $j=\tilde{w}$. (\ref{tv5}) results from (\ref{tvbound}). (\ref{tv6}) follows from 
\begin{equation*}
    \frac{p(x^n|\tilde{w})}{p(x^n)}=\frac{p(x^n|\tilde{w})}{\sum_j p_{\tilde{W}}(j)p_{X^n|\tilde{W}}(x^n|j)}\leq \frac{1}{p(\tilde{w})}.
\end{equation*}
(\ref{tv7}) comes from the fact that $\epsilon$ is chosen in such a way that the typical sets, i.e., $\mathcal{T}^n_\epsilon(p_{X|\tilde{W}}(\cdot|\tilde{w})),\ \forall \tilde{w}$, are disjoint. Hence,
\begin{equation*}
   \textnormal{Pr}\{X^n\in\mathcal{T}^n_{\epsilon}(p_{X|\tilde{W}}(\cdot|\tilde{w}))\}\leq \textnormal{Pr}\{X^n\not\in\mathcal{T}^n_{\epsilon}(p_{X|\tilde{W}}(\cdot|j))\},\ \forall j,\tilde{w}\in\mathcal{\tilde{W}},j\neq \tilde{w}. 
\end{equation*}
Since $\tilde{W}$ is a deterministic function of $W$, we have $|\tilde{\mathcal{W}}|\leq|\mathcal{W}|(<\infty)$. Finally, by noting that as $n$ goes to infinity, (\ref{tv8}) tends to zero, the proof of (\ref{sad}) is complete.

Assume that the mapping $p^*_{Y|X^n}$ is applied to $X^n$, which, as mentioned earlier, satisfies the perfect sample privacy constraints. Hence, from the definition of $I_s$, we have\footnote{From the fact that $W-\tilde{W}-X^n$ and $\tilde{W}-W-X^n$ form a Markov chain, we have $I(W;Y)=I(\tilde{W};Y)$.}
\begin{equation}\label{sandwich}
    I(\tilde{W};Y)\leq I_s(W,X^n)\leq J(W,X^n).
\end{equation}
The claim in (\ref{asymp1}) is proved by showing $\lim_{n\to\infty} J(W,X^n)- I(\tilde{W};Y)=0.$ To this end, we have
\begin{align}
    \lim_{n\to\infty}\sum_{\tilde{w},y}|q(\tilde{w},y)-p^*(\tilde{w},y)|&=\lim_{n\to\infty}\sum_{\tilde{w},y}\bigg|\sum_{x^n}q(\tilde{w},x^n,y)-p^*(\tilde{w},x^n,y)\bigg|\nonumber\\
    &\leq \lim_{n\to\infty}\sum_{\tilde{w},x^n,y}|q(\tilde{w},x^n,y)-p^*(\tilde{w},x^n,y)|\label{tr1}\\
    &=0\label{tr2},
\end{align}
where (\ref{tr1}) follows from the triangle inequality, and (\ref{tr2}) from (\ref{sad}). Hence, the total variation distance between $q_{\tilde{W},Y}$ and $p^*_{\tilde{W},Y}$ vanishes\footnote{Alternatively, (\ref{tr2}) can be proved by considering a deterministic channel that outputs $(W,Y)$, when the tuple $(W,X^n,Y)$ is fed into it. By considering input 1 distributed according to $p^*_{W,X^n,Y}$ and input 2 distributed according to $q_{W,X^n,Y}$, we observe that the TV distance between their corresponding outputs, i.e., $d_{TV}(p^*_{W,Y},q_{W,Y})$, is not greater than $d_{TV}(p^*_{W,X^n,Y},q_{W,X^n,Y})$, which follows from the data processing inequality of f-divergences.} with $n$. Finally, by noting the continuity of mutual information $I(A;B)$ as a functional of $p_{A,B}$, we conclude that 
\begin{equation*}
 \lim_{n\to\infty} J(W,X^n)- I(\tilde{W};Y)=0,   
\end{equation*}
which, in conjunction with (\ref{sandwich}), proves (\ref{asymp1}).
\section{}\label{app2}
Let $\delta\in(0,\epsilon]$ be an arbitrary real number. First, we show that in the Markov chain $X-W-X^n-Y$, we have
\begin{equation}\label{type}
    (1-\delta)p_{X|X^n}(x|x^n)\leq\frac{\sum_{i=1}^np_{X_i|X^n}(x|x^n)}{n}\leq (1+\delta)p_{X|X^n}(x|x^n),\forall x\in\mathcal{X},\forall x^n\in\mathcal{T}_{\delta}^n\left(p_{X|W}(\cdot|w)\right),\forall w.
\end{equation}
This is proved by first noting that the type of a sequence, as defined in (\ref{typedef}), can be alternatively written as
\begin{equation}\label{typedef2}
\pi(x|x^n)=\frac{\sum_{i=1}^np_{X_i|X^n}(x|x^n)}{n},\ \forall x\in\mathcal{X},\forall x^n\in\mathcal{X}^n,
\end{equation}
since the conditional pmf serves as an indicator here. 
Therefore, from (\ref{typic}), we can write
\begin{equation*}
    (1-\delta)p_{X|W}(x|w)\leq\frac{\sum_{i=1}^np_{X_i|X^n}(x|x^n)}{n}\leq (1+\delta)p_{X|W}(x|w),\forall x\in\mathcal{X},\forall x^n\in\mathcal{T}_{\delta}^n\left(p_{X|W}(\cdot|w)\right),\forall w.
\end{equation*}
Moreover, since $\delta\leq \epsilon$, the $\delta$-typical sets $\mathcal{T}_{\delta}^n\left(p_{X|W}(\cdot|w)\right), \forall w$, are disjoint, which results in $p_{W|X^n}(j|x^n)=0, \forall x^n\in\mathcal{T}_{\delta}^n\left(p_{X|W}(\cdot|w)\right), \forall j\neq w,\forall w$. Hence,
\begin{equation*}
   p_{X|X^n}(\cdot|x^n)=\sum_{j} p_{X|W}(\cdot|j)p_{W|X^n}(j|x^n)=p_{X|W}(\cdot|w),\ \forall x^n\in\mathcal{T}_{\delta}^n\left(p_{X|W}(\cdot|w)\right), \forall w.
\end{equation*}
Therefore, (\ref{type}) is proved.

In what follows, we show that in the Markov chain $X-W-X^n-Y$ with $Y\independent X_i,\forall i\in[n]$, we have
\begin{equation}\label{tv0}
    \lim_{n\to\infty}d_{TV}(p_{X,Y},p_X\cdot p_Y)=0.
\end{equation}
Let $\mathbb{T}\triangleq \cup_w\mathcal{T}_{\delta}^n\left(p_{X|W}(\cdot|w)\right)$. We have
\begin{align}
    p_{X,Y}(x,y)&=\sum_{x^n\in\mathcal{X}^n}p_{X,X^n,Y}(x,x^n,y)\nonumber\\
    & = \sum_{x^n\in\mathbb{T}}p_{X,X^n,Y}(x,x^n,y)+\underbrace{\sum_{x^n\not\in\mathbb{T}}p_{X,X^n,Y}(x,x^n,y)}_{\gamma_n(x,y)}\label{gammadef}\\
    &=\sum_{x^n\in\mathbb{T}}p_{X|X^n}(x|x^n)p(x^n,y)+\gamma_n(x,y)\label{lasteq}.
\end{align}
From (\ref{type}), we have
\begin{equation}\label{lasteq2}
   \frac{1}{1+\delta}\frac{\sum_{i=1}^np_{X_i|X^n}(x|x^n)}{n}\leq p_{X|X^n}(x|x^n)\leq  \frac{1}{1-\delta}\frac{\sum_{i=1}^np_{X_i|X^n}(x|x^n)}{n},
\end{equation}
which is valid for $\forall x\in\mathcal{X},\forall x^n\in\mathcal{T}_{\delta}^n\left(p_{X|W}(\cdot|w)\right),\forall w$. Therefore, from (\ref{lasteq}), and (\ref{lasteq2}), we have
\begin{equation}\label{jam}
  \frac{\sum_{i=1}^n\sum_{x^n\in\mathbb{T}}p_{X_i|X^n}(x|x^n)p(x^n,y)}{n(1+\delta)}\leq p_{X,Y}(x,y)-\gamma_n(x,y)\leq  \frac{\sum_{i=1}^n\sum_{x^n\in\mathbb{T}}p_{X_i|X^n}(x|x^n)p(x^n,y)}{n(1-\delta)}.
\end{equation}
We can also write
\begin{align}
    p_{X_i,Y}(x,y)&=\sum_{x^n\in\mathcal{X}^n}p_{X_i,X^n,Y}(x,x^n,y)\nonumber\\
    & = \sum_{x^n\in\mathbb{T}}p_{X_i,X^n,Y}(x,x^n,y)+\underbrace{\sum_{x^n\not\in\mathbb{T}}p_{X_i,X^n,Y}(x,x^n,y)}_{\eta_n(x,y)}.\label{lasteq3}
\end{align}
Therefore, from (\ref{lasteq3}), we can write (\ref{jam}) as
\begin{equation*}
  \frac{\sum_{i=1}^n\left(p_{X_i,Y}(x,y)-\eta_n(x,y)\right)}{n(1+\delta)}\leq p_{X,Y}(x,y)-\gamma_n(x,y)\leq  \frac{\sum_{i=1}^n\left(p_{X_i,Y}(x,y)-\eta_n(x,y)\right)}{n(1-\delta)}, 
\end{equation*}
which further simplifies to
\begin{equation}\label{jam3}
  \frac{p_{X}(x)p_Y(y)-\eta_n(x,y)}{1+\delta}\leq p_{X,Y}(x,y)-\gamma_n(x,y)\leq  \frac{p_{X}(x)p_Y(y)-\eta_n(x,y)}{1-\delta}, 
\end{equation}
since the distribution of $(X_i,X^n,Y)$ is index-invariant, and $p_{X_i}(\cdot)=p_X(\cdot),\forall i\in[n]$.
 By adding $\gamma_n(x,y)-p_X(x)p_Y(Y)$ to (\ref{jam3}), we have
\begin{equation}\label{lasteq4}
    \gamma_n(x,y)-\frac{\eta_n(x,y)}{1+\delta}-\frac{\delta}{1+\delta}p_X(x)p_Y(y)\leq p_{X,Y}(x,y)-p_X(x)p_Y(y)\leq \gamma_n(x,y)-\frac{\eta_n(x,y)}{1-\delta}+\frac{\delta}{1-\delta}p_X(x)p_Y(y).
\end{equation}
From the definitions of $\gamma_n,\eta_n$, in (\ref{gammadef}) and (\ref{lasteq3}), respectively, we have
\begin{equation}
    \sum_{x,y}\gamma_n(x,y)=\sum_{x,y}\eta_n(x,y)=\textnormal{Pr}\{X^n\not\in\mathbb{T}\},
\end{equation}
which, from LLN, tends to zero with $n$. Hence, from (\ref{lasteq4}), we can write
\begin{equation}
   \lim_{n\to\infty} \sum_{x,y}\bigg|p_{X,Y}(x,y)-p_X(x)p_Y(y)\bigg|\leq\frac{\delta}{1-\delta}.
\end{equation}
Since $\delta\in(0,\epsilon]$ was chosen arbitrarily, we must have
\begin{equation}
   \lim_{n\to\infty}d_{TV}(p_{X,Y},p_x\cdot p_Y)=\lim_{n\to\infty} \sum_{x,y}\bigg|p_{X,Y}(x,y)-p_X(x)p_Y(y)\bigg|=0.
\end{equation}
This proves (\ref{tv0}), which states that in the Markov chain $X-W-X^n-Y$, where $Y\independent X_i,\forall i\in[n]$, the pair $(X,Y)$ moves towards independence with $n$. Finally, the continuity of mutual information enables us to conclude that 
\begin{equation}
   \lim_{n\to\infty}I(X;Y)=0.
\end{equation}
\section{}\label{app123}
We divide the proof of the continuity of $C_1(\cdot)$ into two parts: The continuity for $\alpha>0$, and the continuity at $\alpha=0$. Note that only the latter is used in this paper, but the general claim is proved in this appendix.

The first part follows from the concavity of $C_1(\cdot)$. Assume $p^1_{Y|W}$ is the maximizer in $C_1(\alpha_1)$, and $p^2_{Y|W}$ is that in $C_1(\alpha_2)$. Note that the mutual information terms involved in $C_1(\cdot)$ do not depend on the actual realizations of the random variables, but their mass probabilities. Hence, we can assume that $\mathcal{Y}_1\cap\mathcal{Y}_2=\emptyset$, where $\mathcal{Y}_i$ denotes the support of $Y$ induced by the mapping $p^i_{Y|W}, i\in[2].$ Let $Z$ be a binary random variable as a deterministic function of $Y$ that indicates whether it belongs to $\mathcal{Y}_1$ or $\mathcal{Y}_2$. Construct the mapping $p_{Y|W}=\lambda p^1_{Y|W}+ (1-\lambda)p^2_{Y|W}$ for $\lambda\in[0,1]$. Since the pmf of $Z$ does not change by conditioning on $W$, we have $Z\independent W$, and obviously, $Z\independent X$, due to the Markov chain $X-W-Y-Z$. As a result, using the aforementioned mapping $p_{Y|W}$, we have 
\begin{align*}
    I(X;Y)&=I(X;Y,Z)=I(X;Y|Z)=\lambda\alpha_1+(1-\lambda)\alpha_2\\
    I(W;Y)&=I(W;Y,Z)=I(W;Y|Z)=\lambda C_1(\alpha_1)+(1-\lambda)C_1(\alpha_2),
\end{align*}
which proves the concavity of $C_1(\cdot)$, and hence, its continuity over $\alpha>0$.

In order to show the continuity at zero, it is sufficient to consider a sequence of solutions of $C_1(\alpha_n)$, i.e., $\{p^n_{Y|W}\}$, where $\alpha_n$ tends to zero. From Pinsker inequality, we must have $\|p^n(x,y)-p(x)p^n(y)\|_1\to 0$, which happens if and only if at least one of the following statements holds for any $y\in\mathcal{Y}$: 1) $p^n_{W|Y}(\cdot|y)\to\textnormal{Null}(\mathbf{P}_{X|W})$, 2) $p^n_Y(y)\to 0$, which in any case, forces $C_1(\alpha_n)\to C_1(0)$.

\section{}\label{app:heuristic_par}

The joint distribution of $X^n$ and $Y_\text{par}^{n-1}$ can be written as
\begin{equation}\label{eq:joint_structure}
    p_{X^n,Y_\text{par}^{n-1}}(x^n,y_\text{par}^{n-1}) = \prod_{i=1}^n p_{X_i}(x_i) \prod_{j=1}^{n-1}p_{Y_{\text{par},j}|X_j,X_{j+1}}(y_{\text{par},j}|x_j,x_{j+1}).
\end{equation}
Therefore, the indepedency between a given $X_k$ and $Y_\text{par}^{n-1}$, for any $k\in[n]$, can be directly verified by noting that
\begin{align}
I(X_k;Y_\text{par}^{n-1}) &= I(X_k;Y_{\text{par},k-1},Y_{\text{par},k}) \nonumber\\
&= I(X_k;Y_{\text{par},k}|Y_{\text{par},k-1}) \nonumber\\
&= I(Y_{\text{par},k-1},X_k;Y_{\text{par},k}) \nonumber\\
&= 0.\nonumber
\end{align}
Above, the first equality uses the structrure of \eqref{eq:joint_structure}, and the rest of the derivation follows the fact that $I(X_k;Y_{\text{par},k}) = I(X_k;Y_{\text{par},k-1}) = 0$, and that $Y_{\text{par},k-1} - X_k - Y_{\text{par},k}$ form a Markov chain, as shown below:
\begin{align}
    p_{Y_{\text{par},k-1},Y_{\text{par},k}|X_k}(y_{\text{par},k-1},y_{\text{par},k}|x_k) 
    =&\sum_{x_{k-1},x_{k+1}} p_{Y_{\text{par},k-1},Y_{\text{par},k},X_{k-1},X_{k+1}|X_k}(y_{\text{par},k-1},y_{\text{par},k},x_{k-1},x_{k+1}|x_k)\nonumber\\
    =&\sum_{x_{k-1},x_{k+1}} 
    p_{X_{k-1}}(x_{k-1})p_{X_{k+1}}(x_{k+1})
    \prod_{j=k-1}^{k} p_{Y_{\text{par},j}|X_j,X_{j+1}}(y_{\text{par},j}|x_j,x_{j+1}) \nonumber\\
    =& \sum_{x_{k-1}} 
    p_{X_{k-1}}(x_{k-1})p_{Y_{\text{par},k-1}|X_{k-1},X_{k}}(y_{\text{par},{k-1}}|x_{k-1},x_k) \nonumber\\
    &\times \sum_{x_{k+1}}
    p_{X_{k+1}}(x_{k+1})p_{Y_{\text{par},k}|X_k,X_{k+1}}(y_{\text{par},k}|x_k,x_{k+1}) \nonumber \\
    =& \:p_{Y_{\text{par},k-1}|X_k}(y_{\text{par},k-1}|x_k)
    p_{Y_{\text{par},k}|X_k}(y_{\text{par},k}|x_k).
\end{align}

For the second part of the Lemma, it can be shown that
\begin{align}
I(X^n;Y_\text{par}^{n-1}) 
&= I(X^n;Y_\text{par}^{n-1}) \nonumber\\
&= \sum_{j=1}^{n-1} I(X^n;Y_{\text{par},j}|Y^{j-1}_{\text{par}})\nonumber\\
&= \sum_{j=1}^{n-1}  \big[ I(X^n;Y_{j,\text{par}}) - I(Y_{j,\text{par}};Y_\text{par}^{j-1}) \big] \label{eq:001}\\
&= \sum_{j=1}^{n-1}  I(X_j,X_{j+1};Y_{j,\text{par}}) \label{eq:002}\\
&= \sum_{j=1}^{n-1}  I(X_j;Y_{\text{par},j}|X_{j+1}). \nonumber
\end{align}
Above, \eqref{eq:001} uses the fact that $Y_{\text{par},j} - X^n - Y_\text{par}^{j-1}$ is a Markov chain, as can be seen from
\begin{align}
    p_{Y^{j-1}_\text{par},Y_{\text{par},j}|X^n}(y_{\text{par}}^{j-1},y_{\text{par},j}|x^n) 
    &=\sum_{y_{\text{par},j+1},\dots,y_{\text{par},n-1}} 
    \prod_{i=1}^{n-1}p_{Y_{\text{par},i}|X_i,X_{i+1}}(y_{\text{par},i}|x_i,x_{i+1})\nonumber\\
     &=\prod_{i=1}^{j}p_{Y_{\text{par},i}|X_i,X_{i+1}}(y_{\text{par},i}|x_i,x_{i+1})\label{eq:asdvv}\\
     &= p_{Y^{j-1}_\text{par}|X^n}(y_{\text{par}}^{j-1}|x^n) p_{Y_{\text{par},j}|X^n}(y_{\text{par},j}|x^n).\nonumber
\end{align}
Finally, \eqref{eq:002} uses that 
\begin{align}
    p_{Y^{j-1}_{\text{par}},Y_{\text{par},j}|X_j}(y^{j-1}_{\text{par}},y_{\text{par},j}|x_j) 
    =& \sum_{x_i: i\neq j}   p_{Y^{j-1}_\text{par},Y_{\text{par},j}|X^n}(y_{\text{par}}^{j-1},y_{\text{par},j}|x^n) \prod_{k\in[n]: k\neq j}  p_{X_k}(x_k)\nonumber\\
    =& \sum_{x_1,\dots,x_{j-1}} \sum_{x_{j+1}}\prod_{i=1}^{j}p_{Y_{\text{par},i}|X_i,X_{i+1}}(y_{\text{par},i}|x_i,x_{i+1})
    p_{X_{j+1}}(x_{j+1})\prod_{k=1}^{j-1}  p_{X_k}(x_k)
    \label{eq:2222} \\
    =& \sum_{x_{j+1}} p_{Y_{\text{par},j}|X_j,X_{j+1}}(y_{\text{par},j}|x_j,x_{j+1}) p_{X_{j+1}}(x_{j+1}) \nonumber\\
    &\times \sum_{x_1,\dots,x_{j-1}} \prod_{i=1}^{j-1}p_{Y_{\text{par},i}|X_i,X_{i+1}}(y_{\text{par},i}|x_i,x_{i+1})\prod_{k=1}^{j-1}  p_{X_k}(x_k)\nonumber\\
    =& \: p_{Y_{\text{par},j}|X_j}(y_{\text{par},j}|x_j)
    \cdot
    p_{Y^{j-1}_{\text{par}}|X_j}(y^{j-1}_{\text{par}}|x_j)
    . \nonumber
\end{align}
and hence $Y_{\text{par},j} - X_j - Y_\text{par}^{j-1}$ is also a Markov chain and, in turn, $I(Y_{\text{par},j};Y_\text{par}^{j-1}) \leq (Y_{\text{par},j};X_j) =0$. Above, \eqref{eq:2222} is attained using \eqref{eq:asdvv}.

\bibliography{refs.bib}
\bibliographystyle{IEEEtran}
\end{document}